\documentclass{article}
\usepackage{microtype}
\usepackage{graphicx}
\usepackage{subfigure}
\usepackage{booktabs} 

\usepackage[hyphens]{url}
\usepackage[pdftex,
            pdfauthor={},
            pdftitle={Private Lossless Multiple Release},
            pdfsubject={Differential Privacy},
            pdfkeywords={lossless,differential privacy, noise addition},
            pdfcreator={pdflatex}]{hyperref}
            
\usepackage{amsmath, amssymb, amsthm} 
\usepackage{algorithm}
\usepackage{algorithmic}
\usepackage{mathtools}
\usepackage{csquotes}
\usepackage{tabularx}
\usepackage{xspace}

\usepackage[accepted]{icml2025}
\usepackage[capitalize,noabbrev]{cleveref}

\theoremstyle{plain}
\newtheorem{theorem}{Theorem}[section]

\newtheorem{lemma}[theorem]{Lemma}
\newtheorem{corollary}[theorem]{Corollary}

\newtheorem{claim}[theorem]{Claim}

\theoremstyle{definition}

\newtheorem{definition}[theorem]{Definition}

\theoremstyle{remark}

\icmltitlerunning{Private Lossless Multiple Release}

\DeclareMathOperator{\EX}{\mathbb{E}}

\newcommand{\R}{\ensuremath{\mathbb{R}}}

\crefname{lstlisting}{Listing}{Listings}

\newcommand{\cA}{\mathcal{A}}

\newcommand{\cD}{\mathcal{D}}

\newcommand{\cF}{\mathcal{F}}

\newcommand{\cM}{\mathcal{M}}

\newcommand{\cT}{\mathcal{T}}

\newcommand{\cX}{\mathcal{X}}
\newcommand{\cY}{\mathcal{Y}}

\newcommand{\allbold}[1]{{\bfseries\boldmath #1}}
\usepackage{bbold}

\newcommand{\lap}{\ensuremath{\operatorname{Lap}}}
\newcommand{\ber}{\ensuremath{\operatorname{Ber}}} \newcommand{\X}{\ensuremath{\mathbf{x}}}

\DeclareMathOperator*{\var}{Var}
\newcommand{\NN}{\ensuremath{\operatorname{N}}}

\renewcommand{\epsilon}{\varepsilon}

\newcommand{\releaseGeneric}{\ensuremath{\operatorname{Generic Multiple Release}\xspace}}
\newcommand{\releaseGenericS}{\ensuremath{\operatorname{Simplified Generic Multiple Release}\xspace}}
\newcommand{\releaseFf}{\ensuremath{\operatorname{Factorization Multiple Release}\xspace}}  \newcommand{\GLSGH}{\ensuremath{\operatorname{Histogram Gradual Release }}\xspace} \newcommand{\EGLSGH}{\ensuremath{\operatorname{Efficient Histogram Gradual Release}}\xspace}  

\newcommand{\rhomax}{\ensuremath{\rho_{\text{max}}}\xspace}
\newcommand{\zcdp}{$\rho$-zCDP\xspace}
\newcommand{\zcdplong}{$\rho$-zero-concentrated differential privacy (\zcdp)\xspace}

\newcommand{\users}{\ensuremath{S}\xspace}
\newcommand{\user}{\ensuremath{s}\xspace}
\newcommand{\adversarialnodes}{\ensuremath{\users'}\xspace}
\newcommand{\adversarialnode}{\ensuremath{\user'}\xspace}
\newcommand{\userbudget}{\ensuremath{\rho_\user}\xspace}
\newcommand{\adversarialbudget}{\ensuremath{\rho_{\adversarialnode}}\xspace}

\renewcommand\paragraph[1]{{{\textbf{#1.}}}}

\newcommand{\MM}{\mathcal{M}}
\newcommand{\MMM}{\textsc{M}}

\newcommand{\DD}{\mathcal{D}}
\newcommand{\CC}{\mathcal{C}}

\newcommand{\XX}{\mathcal{X}}
\newcommand{\YY}{\mathcal{Y}}

\newcounter{casenum}
\newenvironment{caseof}{\setcounter{casenum}{1}}{\vskip.5\baselineskip}
\newcommand{\case}[2]{\vskip.5\baselineskip\par\noindent {\bfseries Case \arabic{casenum}:} #1\\#2\addtocounter{casenum}{1}}

\newcommand{\IAN}{independent additive noise\xspace}

\definecolor{orangeish}{HTML}{f1a340} \definecolor{grayish}{HTML}{f7f7f7} \definecolor{purpleish}{HTML}{998ec3} \definecolor{blueish}{HTML}{004488}
    \definecolor{darkgreen}{RGB}{2,100,64} \definecolor{lightgreen}{HTML}{b2f2bb}
    \definecolor{lightblueish}{RGB}{230,244,255}
    \definecolor{gold}{rgb}{0.83, 0.69, 0.22}
    
\PassOptionsToPackage{table,xcdraw,svgnames,dvipsnames}{xcolor} 

\makeatletter
\newcommand{\PARAMETERS}{\item[\algorithmicparameters]}
\newcommand{\algorithmicparameters}{\textbf{Parameters:}}
\newcommand{\INPUTS}{\item[\algorithmicinputs]}
\newcommand{\algorithmicinputs}{\textbf{Inputs:}}
\makeatother

\usepackage{tikzit}

\tikzstyle{BLACK}=[draw=black, shape=circle, fill=black, inner sep=3pt]
\tikzstyle{DOM}=[fill={rgb,255: red,122; green,0; blue,42}, draw=black, shape=circle, inner sep=3pt]
\tikzstyle{NONE}=[fill={rgb,255: red,247; green,33; blue,212}, draw=black, shape=circle, inner sep=3pt]
\tikzstyle{NTWO}=[fill={rgb,255: red,0; green,177; blue,219}, draw=black, shape=circle, inner sep=3pt]
\tikzstyle{NTHR}=[fill={rgb,255: red,140; green,143; blue,14}, draw=black, shape=circle, inner sep=3pt]
\tikzstyle{GRAYN}=[fill={rgb,255: red,128; green,128; blue,128}, draw=black, shape=circle, inner sep=3pt]
\tikzstyle{BRGRAY}=[fill={rgb,255: red,236; green,236; blue,236}, draw=none, shape=circle, inner sep=3pt]
\tikzstyle{TEXTSTD}=[fill=white, draw=black, shape=rectangle, tikzit shape=rectangle, text width=3cm, rounded corners]
\tikzstyle{TEXTHIGH}=[fill={rgb,255: red,231; green,245; blue,255}, draw=black, shape=rectangle, tikzit shape=rectangle, text width=3cm, rounded corners]
\tikzstyle{TEXTHIGHGRAY}=[fill={rgb,255: red,220; green,220; blue,220}, draw=black, shape=rectangle, tikzit shape=rectangle, text width=3cm, rounded corners]
\tikzstyle{TEXTHIGHGREEN}=[fill={rgb,255: red,149; green,242; blue,145}, draw=black, shape=rectangle, tikzit shape=rectangle, text width=3.4cm, rounded corners]
\tikzstyle{TEXTHIGHYELLOW}=[fill={rgb,255: red,250; green,250; blue,160}, draw=black, shape=rectangle, tikzit shape=rectangle, text width=3cm, rounded corners]

\tikzstyle{EDGE}=[-, fill=none, line width=2pt]
\tikzstyle{BLUE}=[-, draw={rgb,255: red,0; green,101; blue,189}, line width=1.5pt]
\tikzstyle{GRAY}=[-, fill={rgb,255: red,128; green,128; blue,128}]
\tikzstyle{DARKGREEN}=[-, fill=none, draw={rgb,255: red,19; green,100; blue,13}, line width=1.25pt]
\tikzstyle{big dash}=[-, dashed, dash pattern=on 1mm off 2mm, fill={rgb,255: red,178; green,255; blue,253}]
\tikzstyle{big dash thick}=[-, thick, dashed, dash pattern=on 4mm off 2mm, fill={rgb,255: red,178; green,255; blue,253}]
\tikzstyle{new edge style 0}=[-, fill={rgb,255: red,70; green,255; blue,243}, line width=1.5pt]
\tikzstyle{FILLGREEN}=[-, fill={rgb,255: red,144; green,237; blue,134}, line width=1pt]
\tikzstyle{FILLRED}=[-, fill={rgb,255: red,237; green,59; blue,36}, line width=1pt]
\tikzstyle{FILLBLUE}=[-, fill={rgb,255: red,162; green,232; blue,230}, line width=1pt]
\tikzstyle{FILLPURPLE}=[-, fill={rgb,255: red,255; green,233; blue,239}]
\tikzstyle{FILLDARKBLUE}=[-, fill={rgb,255: red,238; green,230; blue,255}]
\tikzstyle{FILLDARKBLUEINVISIBLE}=[-, draw=none, fill={rgb,255: red,238; green,230; blue,255}]
\tikzstyle{FILLGREENINVISBLE}=[-, draw=none, fill={rgb,255: red,232; green,255; blue,207}]
\tikzstyle{ULTRAGRAY}=[-, draw={rgb,255: red,150; green,150; blue,150}]
\tikzstyle{EDGEDASHED}=[-, dashed, fill=none, line width=1pt]
\tikzstyle{simple directed edge}=[draw=black, thick, line width=1pt, ->]
\tikzstyle{simple dashed directed edge}=[->, draw=black, dashed]
\tikzstyle{SLIGHTLYDASHED}=[-, dotted, fill=none, draw={rgb,255: red,128; green,128; blue,128}]
\tikzstyle{simple dashed directed edge}=[->, draw=black, line width=0.5mm, dashed]
\tikzstyle{simple dashed directed edge not thick}=[->, draw=black, line width=0.45mm, dashed]
\tikzstyle{simple}=[->, draw=black, thick]
\tikzstyle{noise arrow}=[draw={rgb,255: red,128; green,44; blue,127}, fill=none, ->, line width=0.4mm]
\tikzstyle{invis}=[-, draw=none]
\tikzstyle{red-edge}=[-, draw={rgb,255: red,191; green,0; blue,64}, fill=none, line width=1.5pt]
\tikzstyle{thick}=[-, fill=none, line width=1.25pt]
\tikzstyle{noise arrow stopper}=[-, line width=0.35mm, draw={rgb,255: red,116; green,66; blue,128}]
\tikzstyle{black dotted}=[-, dotted, line width=1pt]
\tikzstyle{new edge style 1}=[line width=0.5mm, draw=black, ->, fill=none]
\tikzstyle{thick2pt}=[-, line width=1pt]
\tikzstyle{purple arrow}=[line width=1pt, draw={rgb,255: red,128; green,0; blue,128}, ->]
\tikzstyle{new edge style 2}=[line width=1pt, ->, draw={rgb,255: red,0; green,101; blue,189}]
\tikzstyle{FILLBLUEINVISBLE}=[-, fill={rgb,255: red,198; green,238; blue,255}, draw=none]
  
\begin{document}

\twocolumn[
\icmltitle{Private Lossless Multiple Release}

\icmlsetsymbol{equal}{*}

\begin{icmlauthorlist}
    \icmlauthor{Joel Daniel Andersson}{barc,diku}
    \icmlauthor{Lukas Retschmeier}{barc,diku}
    \icmlauthor{Boel Nelson}{diku}
    \icmlauthor{Rasmus Pagh}{barc,diku}
\end{icmlauthorlist}

\icmlaffiliation{barc}{Basic Algorithms Research Copenhagen (BARC), Denmark}
\icmlaffiliation{diku}{University of Copenhagen, Denmark}

\icmlcorrespondingauthor{Joel Daniel Andersson}{jda@di.ku.dk}
\icmlcorrespondingauthor{Lukas Retschmeier}{lure@di.ku.dk}
\icmlcorrespondingauthor{Boel Nelson}{bn@di.ku.dk}
\icmlcorrespondingauthor{Rasmus Pagh}{pagh@di.ku.dk}

\icmlkeywords{Differential Privacy, Algorithms}

\vskip 0.3in
]

\printAffiliationsAndNotice{}  

\begin{abstract}
Koufogiannis et al.~(2016) showed a \emph{gradual release} result for Laplace noise-based differentially private mechanisms: given an \(\varepsilon\)-DP release, a new release with privacy parameter \(\varepsilon' > \varepsilon\) can be computed such that the combined privacy loss of both releases is at most \(\varepsilon'\) and the distribution of the latter is the same as a single release with parameter \(\varepsilon'\).
They also showed gradual release techniques for Gaussian noise, later also explored by Whitehouse et al.~(2022).

In this paper, we consider a more general \emph{multiple release} setting in which analysts hold private releases with different privacy parameters corresponding to different access/trust levels.
These releases are determined one by one, with privacy parameters in arbitrary order. 
A multiple release is \emph{lossless} if having access to a subset~\(S\) of the releases has the same privacy guarantee as the least private release in \(S\), and each release has the same distribution as a single release with the same privacy parameter.
Our main result is that lossless multiple release is possible for a large class of additive noise mechanisms.
For the Gaussian mechanism we give a simple method for lossless multiple release with a short, self-contained analysis that does not require knowledge of the mathematics of Brownian motion.
We also present lossless multiple release for the Laplace and Poisson mechanisms.
Finally, we consider how to efficiently do gradual release of sparse histograms, and present a mechanism with running time independent of the number of dimensions.
\end{abstract}

\section{Introduction}\label{sec:introduction}

\emph{Differential privacy}~\cite{dwork_calibrating_2006} is a statistical notion that provides provable privacy guarantees. 
Differentially private (DP) algorithms typically introduce inaccuracy through noise to achieve privacy, and the resulting privacy-accuracy trade-off is the key object of study in the area. 
Of specific interest is the \emph{privacy budget} that determines how much information the output of a differentially private algorithm may reveal about its inputs---a smaller budget means more private but less accurate results.

\paragraph{Motivation} In deployments of differential privacy, it may be hard to determine the appropriate privacy budget to grant an analyst since it depends on trust and accuracy assumptions that may change over time. 
A system may also have different security clearance levels---a data analyst might have a higher clearance level than a developer, but both might require access to \emph{some} statistics. 
Similarly, a company that wants to release, for example, user statistics could make an accurate release for their own data analysts, a less accurate release for external consultants, and include an even less accurate release in a report for shareholders or other external actors. 
A different example setting is users that want to sell their data on data markets: users could let the accuracy of the release depend on how much they are paid, and use different budgets for different releases.

Another usage scenario, relevant in distributed or federated settings, is \emph{local differential privacy} where the data owner could be sharing a differentially private function of their data with multiple servers, and the set of servers might change over time.
Here the privacy budget might depend on the server: for example, a patient may trust their local hospital more than their national hospitals, but might not trust the hospitals not to collude by sharing data among themselves.

\paragraph{Multiple releases}
These scenarios motivate creating \emph{multiple releases} with different privacy budgets aimed at different analysts.
However, these releases should be \emph{coordinated} such that a group of analysts who combine their information do not gain more knowledge about the input than the most knowledgeable member of the group.
This kind of collusion resilience was first studied by~\citet{XiaoTC09} with a non-DP privacy objective---we refer to their work for additional motivation.

A related aspect is that we may want to provide an analyst with a less private, more accurate release after the trust we place in them increases.
In this case, we want the accuracy of the latest release to match the accuracy that can be obtained, given the combined privacy budget of both releases. 
That is, no additional cost should be incurred for making two releases rather than one. 
For example, an external consultant that later gets employed directly by the company should be able to get access to the more accurate release without constituting a privacy violation or requiring an increased privacy budget to reach the same accuracy.

\paragraph{Baseline}
It is possible to create multiple releases for \emph{any} differentially private mechanism if we are willing to increase the privacy budget by a constant factor.
In particular, we can create a sequence of independent releases with geometrically increasing privacy parameters, referred to as ``$\varepsilon$-doubling'' by~\citet{LigettNRWW17}, and provide each analyst with the most accurate release they are entitled to.
Using composition results, the combined privacy budget for the information given to a set of analysts is a constant factor from the highest budget of a single analyst in the group.
In this paper, we study how to make multiple releases in a \emph{lossless} way without accuracy or privacy penalty.

\subsection{Related Work}
\citet*{koufogiannis_2016} introduced the concept of \emph{gradual release}, which makes it possible to increase the privacy budget with \emph{no loss} in accuracy.
They also considered privacy tightening for the Laplace mechanism, where successive releases are increasingly private.
In the journal version of the paper, they also introduce gradual release for the Gaussian mechanism under approximate DP which is based on the machinery of Brownian motion.

This technique was later applied in a noise reduction framework by~\citet*{LigettNRWW17} with the goal of producing mechanisms with ex-post privacy, i.e., where the privacy budget is set based on what is required to achieve a desired accuracy level.
Follow-up work by~\citet*{whitehouse2022} gave results for noise reduction under approximate DP using Brownian motion.
These works were motivated by work on privacy filters and odometers~\citep{RogersVRU16}, keeping track of privacy budgets over time, rather than multiple release settings.
Recently, \citet*{pan2024randomized} demonstrated lossless gradual release for randomized response.

Similar research questions have been investigated outside of the DP literature.
\citet*{XiaoTC09} studied releasing a sensitive dataset where each element is kept with probability $p$, and otherwise sampled uniformly from the universe.
\citet*{LiCLZ12} considered privatizing data by additive Gaussian noise of scale $\sigma$.
Both works dealt with arbitrary sequences of parameters (probabilities $p$ or noise scalings $\sigma$), and demonstrated how to correlate releases to guarantee that (1) each release matches the single-release case and, (2) limiting the sensitive information derived from combining releases.
Our method for adaptively producing Gaussian releases deviates from~\cite{LiCLZ12}, in that ours does not need to maintain any covariance matrix for past releases.
This makes our approach more time- and space-efficient.

\begin{figure}[t]
    \centering
    \scalebox{1.0}{
        \begin{tikzpicture}
	\begin{pgfonlayer}{nodelayer}
		\node [style=none] (0) at (8, 0) {};
		\node [style=none] (2) at (8.75, 0) {};
		\node [style=none] (4) at (9.25, 0) {$\rho$};
		\node [style=none] (17) at (-3.25, 0.25) {};
		\node [style=none] (19) at (-3.25, -0.25) {};
		\node [style=none] (26) at (2, 0.75) {$\rho_k$};
		\node [style=none] (27) at (6.25, 0.75) {$\rho_{k+1}$};
		\node [style=none] (31) at (3.5, 1) {$\rho$};
		\node [style=none] (34) at (3.5, 0.5) {};
		\node [style=none] (35) at (3.5, -0.75) {};
		\node [style=none] (44) at (-1.75, 0.75) {$\rho_{1}$};
		\node [style=none] (45) at (0.25, 0) {};
		\node [style=none] (46) at (-3.25, 0) {};
		\node [style=none] (47) at (1.5, 0) {};
		\node [style=none] (48) at (0.75, 0) {};
		\node [style=none] (49) at (7.75, 0) {};
		\node [style=none] (50) at (6.5, 0) {};
		\node [style=none] (51) at (7, 0) {};
		\node [style=none] (52) at (1.75, 0) {};
		\node [style=none] (53) at (6.25, 0.25) {};
		\node [style=none] (54) at (6.25, -0.25) {};
		\node [style=none] (55) at (2, -0.25) {};
		\node [style=none] (56) at (2, 0.25) {};
		\node [style=none] (57) at (-1.75, 0.25) {};
		\node [style=none] (58) at (-1.75, -0.5) {};
		\node [style=none] (59) at (-1.75, -1.25) {$Y_{\rho_1}$};
		\node [style=none] (63) at (2, 0.25) {};
		\node [style=none] (65) at (2, -1.25) {$Y_{\rho_k}$};
		\node [style=none] (66) at (6.25, 0.25) {};
		\node [style=none] (68) at (6.25, -1.25) {$Y_{\rho_{k+1}}$};
		\node [style=none] (69) at (3.5, -1.25) {$Y_\rho$};
		\node [style=none] (70) at (2, 0.25) {};
		\node [style=none] (71) at (2, -0.5) {};
		\node [style=none] (72) at (6.25, 0.25) {};
		\node [style=none] (73) at (6.25, -0.5) {};
		\node [style=none] (74) at (5.25, -1.25) {};
		\node [style=none] (75) at (4, -1.25) {};
		\node [style=none] (76) at (2.5, -1.25) {};
		\node [style=none] (77) at (3, -1.25) {};
		\node [style=none] (78) at (8.25, 0.25) {};
		\node [style=none] (79) at (8.25, -0.5) {};
		\node [style=none] (80) at (8.25, -1.25) {$Y_{\rho_\infty}$};
		\node [style=none] (81) at (8.25, 0.75) {$\rho_{\infty}$};
		\node [style=none] (82) at (-0.25, 0.75) {$\rho_{2}$};
		\node [style=none] (83) at (-0.25, 0.25) {};
		\node [style=none] (84) at (-0.25, -0.5) {};
		\node [style=none] (85) at (-0.25, -1.25) {$Y_{\rho_2}$};
		\node [style=none] (86) at (-3, 0.75) {$\rho_0$};
		\node [style=none] (87) at (-3, 0.25) {};
		\node [style=none] (88) at (-3, -0.5) {};
		\node [style=none] (89) at (-3, -1.25) {$Y_{\rho_0}$};
		\node [style=none] (90) at (-3, 1.5) {{\footnotesize most private}};
		\node [style=none] (91) at (8.25, 1.5) {{\footnotesize least private}};
	\end{pgfonlayer}
	\begin{pgfonlayer}{edgelayer}
		\draw [style=FILLBLUEINVISBLE] (53.center)
			 to (54.center)
			 to (55.center)
			 to (56.center)
			 to cycle;
		\draw [style=simple directed edge] (0.center) to (2.center);
		\draw [style=EDGE] (17.center) to (19.center);
		\draw [style=new edge style 2] (34.center) to (35.center);
		\draw [style=thick2pt] (45.center) to (46.center);
		\draw [style=EDGEDASHED] (48.center) to (47.center);
		\draw [style=thick2pt] (52.center) to (50.center);
		\draw [style=EDGEDASHED] (51.center) to (49.center);
		\draw [style=simple directed edge] (57.center) to (58.center);
		\draw [style=simple directed edge] (70.center) to (71.center);
		\draw [style=simple directed edge] (72.center) to (73.center);
		\draw [style=simple directed edge] (74.center) to (75.center);
		\draw [style=simple directed edge] (76.center) to (77.center);
		\draw [style=purple arrow] (78.center) to (79.center);
		\draw [style=simple directed edge] (83.center) to (84.center);
		\draw [style=simple directed edge] (87.center) to (88.center);
	\end{pgfonlayer}
\end{tikzpicture}
    }
   \caption{
   The idea behind \emph{lossless multiple release}.
   For concreteness we consider additive noise mechanism and zero-concentrated differential privacy.
   Each $Y_{\rho_i}$ denotes a noisy estimate, and to release a new estimate with $\rho > 0$, we can combine the adjacent estimates for $\rho_k$ and $\rho_{k+1}$ together with some fresh noise to obtain a new release $Y_\rho$ that is exactly $\rho$-zCDP.
   Note that these estimates do not need to be strictly increasing (or decreasing) but can be released in \emph{any order}. 
   Furthermore, releasing any subset of these estimates is exactly $\max(S)$-zCDP, where $S$ is the set of privacy parameters.}
\label{fig:the-idea}
\end{figure}
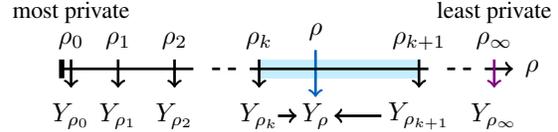

\subsection{Basic Technique}
We illustrate our framework in the setting of additive Gaussian noise and \zcdplong \footnote{see definitions in \cref{sec:preliminaries}}.
We want to release a private estimate of a real-valued query ${f: \mathcal{X} \rightarrow \R}$ with $\ell_2$-sensitivity $\Delta_2 = 1$ using the Gaussian mechanism.
Now consider the simple case where one wants to privately release an estimate with two privacy levels ${\rho < \rho'}$.
Then releasing ${Y_\rho = f(\X) + \NN(0, \frac{1}{2\rho})}$ together with ${Y_{\rho'-\rho} = f(\X) + \NN(0, \frac{1}{2(\rho'-\rho)})}$ is $\rho'$-zCDP by composition.
Now observe that we can combine these estimates to produce a better estimate using inverse variance weighting: $Y = \frac{\rho}{\rho+(\rho'-\rho)} Y_\rho + \frac{\rho' - \rho}{\rho+(\rho' - \rho)}Y_{\rho'-\rho}$ yields \emph{exactly} the same utility as a single release under $\rho'$-zCDP.
That is, instead of using a privacy budget of $\rho + \rho'$ for independently releasing both estimates, the overall budget spent is just the maximum of both, which is $\rho'$.

To do multiple lossless releases for more general additive noise mechanisms, it turns out that one can always combine existing releases with fresh noise as illustrated in \Cref{fig:the-idea}.
In fact, to make any new release with privacy parameter $\rho$, it suffices to have saved the two ``adjacent'' releases whose privacy parameters are closest to $\rho$.
Our approach applies more generally to a class of (independent) additive noise mechanisms that support gradual release, but the releases can be made in any order, and we do not require the set of releases to be known in advance. 

\paragraph{Our contributions} 
We introduce a framework for \emph{lossless multiple release}, generalizing past work on gradual release.
In addition to allowing for making multiple private releases and having the privacy loss only scale with the least private release, we impose no specific ordering on the privacy parameters used.
Furthermore, we formalize a general theorem (\Cref{sec:meta-theorem}) showing that lossless multiple release is possible for a large class of mechanisms based on adding i.i.d. noise to each coordinate whose distribution satisfies a \emph{convolution preorder}, as defined next.
\begin{definition}[Convolution preorder]\label{def:decomposable}
    A family of real-valued distributions $\DD(\rho)$ parameterized by $\rho\in\R_+$ is said to satisfy a \emph{convolution preorder}, if given $\DD(\rho_1)$ and $\DD(\rho_2)$ for $\rho_1 < \rho_2$, there exists a distribution $\CC(\rho_2, \rho_1)$ such that $\DD(\rho_2) * \CC(\rho_2, \rho_1) = \DD(\rho_1)$.
\end{definition}
This relation is a natural way of stating that the distributions become more noisy as the parameter $\rho$ gets smaller.
We are unaware of any existing term in the literature but the definition is closely related to \emph{convolution order}~\cite{stochasticorders2007}. 
Examples of noise distributions satisfying Definition~\ref{def:decomposable} include the Laplace, Poisson, and the Gaussian mechanisms, parameterized by a decreasing function of their variance.

\begin{theorem}[Meta theorem, informal]\label{th:meta-informal}
    Let ${\cA_{f,\rho}:\cX \to \R^d}$ be a mechanism that adds independent, identically distributed noise to coordinates of a function $f:\cX \to \R^d$.
    If the noise distribution of $\cA_{f, \rho}$ satisfies \emph{convolution preorder}, then there exists an algorithm enabling \emph{lossless multiple release}.
    Also, any invertible post-processing of $\cA_{f,\rho}$ preserves this property.
\end{theorem}

For the Gaussian mechanism in particular, we show how to get lossless multiple release by only using basic properties of Gaussians, and we show similar results for the Laplace and Poisson mechanisms from first principles.
We give concrete instantiations of \emph{Gaussian sparse histograms} (\cref{sec:sparse-histograms}) and \emph{factorization mechanisms} (\cref{sec:factorization-mechanism}).

 \section{Threat Model and Goals}\label{sec:threat-model}

Our setting is a multi-user system with a set of analysts/users~\users, all able to query the same dataset using a differentially private mechanism. 
User \user has their own security clearance level with corresponding privacy budget $\userbudget$. 
This setting allows for multi-level security where each user belongs to a security clearance level, like in the classic Bell-LaPadula model~\cite{bell_secure_1976}, where users with high clearance levels have higher values of $\rho$. 
A common example of such clearance levels is using increasing levels with labels such as \emph{public}, \emph{restricted}, \emph{confidential}, and \emph{top secret}. 
The user with the highest clearance in the system has a privacy budget of $\underset{\user \in \users}{\max}(\userbudget)$, which we denote \rhomax.

We consider an adversary who gains partial knowledge, that is, an adversary that sees \emph{some} of the releases. Our goal is to design a differentially private mechanism such that an adversary with access to releases from some users ${\adversarialnodes \subseteq \users}$, can learn \emph{at most} what they could have learned from a release with privacy parameter $\underset{\adversarialnode \in \adversarialnodes}{\max}(\adversarialbudget)$.
This means that even by compromising or colluding with more users, the adversary's knowledge may not increase. 
In case an adversary observes all releases (e.g., by compromising all users), the privacy loss would be bounded by \rhomax. 

However, our goal is to design a mechanism where the releases are \emph{lossless} in the sense that the noise distribution from multiple releases with a combined budget \rhomax would be indistinguishable from one single release with the privacy budget \rhomax. 
In other words, the privacy loss should be determined by \rhomax, while independent releases would usually have privacy parameter $\underset{\user \in \users}{\sum}\userbudget$ due to composition.
 
\section{Gaussian Lossless Multiple Release}\label{sec:gradual-release}

Extending the work in~\cite{koufogiannis_2016,LiCLZ12}, we demonstrate next that in the case of the Gaussian mechanism, providing lossless multiple release is clean and follows immediately from simple properties of Gaussians.
Throughout the paper, we use the symbol $Y$ to be a random variable that depends on the private dataset and $Z$ to be one that does not.

We first consider the gradual release setting where an increasingly accurate estimate is released, and the overall privacy loss is determined solely by the latest, least private release.
In \Cref{lemma:gaussian_seq_ooo}, we drop this restriction, allowing releases in any order while still guaranteeing that the overall privacy guarantee is the maximum $\rho$ value provided.
For simplicity, we consider the one-dimensional case, where we want to release some query $f:\mathcal{X} \rightarrow \R$.
The $d$-dimensional setting is handled by sampling each coordinate independently.

\paragraph{Getting started}
Foreshadowing the usage of $\rho$-zCDP, we initially use $1/(2\rho)$ for denoting the variance of a Gaussian.
Our inquiry starts with a basic observation about inverse-variance weighting of Gaussians.
\begin{lemma}\label{lemma:inverse-variance-weighting-is-nice}
    Let $Y_\rho\sim\NN(\beta, \frac{1}{2\rho})$ and $Y_{\rho'}\sim\NN(\beta, \frac{1}{2\rho'})$ where $\beta\in\mathbb{R}$ and $\rho, \rho' > 0$.
    Then
    \begin{equation*}
        Y = \tfrac{\rho}{\rho + \rho'}Y_{\rho} + \tfrac{\rho'}{\rho + \rho'}Y_{\rho'} \sim\NN\bigg(\beta, \tfrac{1}{2(\rho + \rho')}\bigg)\,.
    \end{equation*}
\end{lemma}
\begin{proof}
   The mean of $Y$ is immediate by the fact that $Y$ is a weighted-average of $Y_{\rho}$ and $Y_{\rho'}$, both of mean $\beta$.
   For the variance, direct computation yields:
   \begin{equation*}
       \var[Y]=\frac{\rho^2/(2\rho) + \rho'^2/(2\rho')}{(\rho_1 + \rho')^2} = \frac{1}{2(\rho + \rho')}\,.\end{equation*}
   As the sum of two Gaussians is itself Gaussian, we are done.
\end{proof}

The essence of what is being claimed is that a Gaussian of variance $\frac{1}{2\rho}$ and another Gaussian of variance $\frac{1}{2\rho'}$ with the same mean can be combined into a new Gaussian with variance $\frac{1}{2(\rho + \rho')}$ and the same mean.
Inspired by this, we can repeatedly invoke the lemma for the following result.
\begin{lemma}\label{lemma:gaussian_seq}
    Given $0 < \rho_1 < \dots < \rho_m$ and $\beta\in\R$ define $Y_{\rho_1}, \dots, Y_{\rho_m}\in\R$ where $Y_{\rho_1} \sim \NN(\beta,\frac{1}{\rho_1})$, and for $k > 1$:
    \begin{equation*}
    Y_{\rho_{k+1}} = \tfrac{\rho_{k}}{\rho_{k+1}} Y_{\rho_k} + \tfrac{\rho_{k+1} - \rho_{k}}{\rho_{k+1}} \cdot \NN\bigg(\beta, \frac{1}{2(\rho_{k+1} - \rho_{k})}\bigg)\,.
    \end{equation*}
    Then for any $i\in[m] : Y_{\rho_i}\sim\NN(\beta, \frac{1}{2\rho_i})$ and for any ${j\in[m]} : \mathrm{Cov}(Y_{\rho_i}, Y_{\rho_j}) = \frac{1}{2\max(\rho_i, \rho_j)}$.
\end{lemma}
\begin{proof}
   We first show the distribution of $Y_{\rho_k}$ by induction on~$k$.
   Assume that $Y_{\rho_k}\sim\NN(\beta, \frac{1}{\rho_k})$, which is true for the base case of $k=1$.
   Note that the inductive step follows immediately from invoking~\cref{lemma:inverse-variance-weighting-is-nice}.
   For the covariance, note that we can expand the expressions inside the covariance and \enquote{throw away} the independent noise added in each recurrence.
   Assuming ${i \leq j}$ we get ${\mathrm{Cov}(Y_{\rho_i}, Y_{\rho_j}) = \mathrm{Cov}(Y_{\rho_i}, \frac{\rho_i}{\rho_j}Y_{\rho_i}) = \frac{\rho_i}{\rho_j}\var[Y_{\rho_i}] = \frac{1}{2\rho_j}}$, implying the stated covariance.
\end{proof}

\paragraph{Releases in arbitrary order}
The Gaussian sequence in \Cref{lemma:gaussian_seq} will be the basis for our lossless multiple release version of the Gaussian mechanism.
What the lemma does not address is generating the sequence in arbitrary order: 
Statically, given the full sequence $(\rho_k)_{k\in[m]}$, we can generate the Gaussians, but what if we receive them one-by-one and in arbitrary order?
We will address this next.
\begin{lemma}\label{lemma:gaussian_seq_ooo}
    For $\rho_\infty\in\R_+$ (possibly $\rho_\infty=\infty$) and $\beta\in\R$, let $M = \{(0, \infty), (\rho_\infty, Y_{\rho_\infty})\}$ where $Y_{\rho_\infty}\in\NN(\beta, \frac{1}{2\rho_\infty})$.
    Consider a finite subset $S\subset (0, \rho_\infty)$ and the following process that runs for $\lvert S \rvert$ iterations:
    \begin{enumerate}
        \item Pick an arbitrary $\rho\in S$ and delete it from $S$;
        \item Let $(\rho_l, Y_{\rho_l}), (\rho_r, Y_{\rho_r})\in M$ where $\rho\in(\rho_l, \rho_r)$ and $\rho_r - \rho_l$ is minimal;
        \item Sample $Z\sim \NN\bigg(0, \tfrac{(1-\rho_l / \rho)(1 / \rho - 1 / \rho_r)}{2(1-\rho_l/\rho_r)}\bigg)$, let 
        \begin{align*}
            Y_\rho & =  \tfrac{1-\rho_l / \rho}{1 - \rho_l / \rho_r} Y_{\rho_r} + \tfrac{\rho_l/\rho - \rho_l/\rho_r}{1 - \rho_l / \rho_r}Y_{\rho_l} + Z,
        \end{align*}
        and add $(\rho, Y_\rho)$ to $M$.
    \end{enumerate}
    Then the sequence of random variables generated by the process has the same distribution as described in \Cref{lemma:gaussian_seq}.
\end{lemma}
\begin{proof}[Proof sketch.]
    The argument is inductive.
    Under the hypothesis that all values generated up to the given point have the distribution described by \Cref{lemma:gaussian_seq}, we argue that the newly generated value does too.
    The argument considers the four different cases for $\rho_l, \rho_r$, e.g., $\rho_l = 0, \rho_r \neq \rho_\infty$.
As the proof involves tedious computation we refer the reader to \Cref{appendix:omitted-proofs} for the formal proof.
\end{proof}

\paragraph{Formalizing lossless multiple release}
\Cref{lemma:gaussian_seq_ooo} will constitute the basis for our implementation of lossless multiple release, but we have yet to formally define this notion.
We do so next.
\begin{definition}[Lossless multiple release]\label{def:lossless}
    Let ${\mathcal{M}_\rho : \mathcal{X} \to \mathcal{Y}}$ be a family of mechanisms on a domain $\mathcal{X}$, indexed by a privacy parameter \(\rho\in \R_+\).
    We say that $\textsc{M}: \mathcal{X} \times \R_+ \to \mathcal{Y}$ implements $\mathcal{M}_\rho$ with \emph{lossless multiple release} if for every $x\in \mathcal{X}$ it satisfies:
    \begin{enumerate}
        \item \(\forall \rho\): \(\textsc{M}(x,\rho)\) and \(\mathcal{M}_{\rho}(x)\) are identically distributed.
        \item For every finite subset $S \subset \R_+$, processed in arbitrary order by $\MMM$, and $y\in\mathcal{Y}$, conditioned on $\textsc{M}(\X,\max(S)) = y$ the joint distribution of $(\textsc{M}(\X,\rho))_{\rho \in S}$ is uniquely determined by $y$ and $S$.
    \end{enumerate}
\end{definition}
Functionally, a mechanism meets the definition if its outputs can be correlated such that for any set of outputs, their joint distribution can be viewed as (randomized) post-processing of the least private release.
An implementation necessarily has to store information about releases that have been made, i.e., the sequence of inputs to $\MMM$ and the corresponding outputs, to fulfill the requirement that releases for different privacy parameters are correlated. 
When $\MMM$ only supports outputting releases for a sequence of \emph{increasing} privacy parameters, we call it \emph{lossless gradual release}.

Having stated the definition for lossless multiple release, consider \Cref{alg:gaussian-multi-release}.
It implements \Cref{lemma:gaussian_seq_ooo}, and the idea is visualized in \cref{fig:the-idea}.

\begin{algorithm}[t]\caption{$\operatorname{Gaussian Multiple Release}$}
\label{alg:gaussian-multi-release}
\begin{algorithmic}[1]
\PARAMETERS $\ell_2$ sensitivity $\Delta_2$
\INPUTS Set of releases $M$, privacy parameter $\rho$ 
\STATE ~~~Find ${(\rho_k, Y_{\rho_k}), (\rho_{k+1}, Y_{\rho_{k+1}})\in M}$ such that 
\STATE ~~~~~~$\rho\in [\rho_k, \rho_{k+1})$ and $\forall (\rho', \cdot)\in M: \rho' \notin (\rho_k, \rho_{k+1})$\label{line:closest-rhos}
\STATE ~~~Sample $Z_\rho \sim\NN(0, \Delta_2^2 \cdot \frac{(1-\rho_k/\rho)\cdot (1/\rho-1/\rho_{k+1})}{2(1 - \rho_k/\rho_{k+1})})$\label{line:sample-new}
\STATE~~~Let $Y_\rho := Z_\rho + \frac{(1 -\rho_k/\rho)Y_{\rho_{k+1}} + (\rho_k/\rho - \rho_k/\rho_{k+1}) Y_{\rho_k}}{1 - \rho_k / \rho_{k+1}}$ \label{line:estimate-new}
\STATE ~~~Add $M = M \cup \{(\rho, Y_\rho)\}$\label{line:store-new}
\STATE ~~~\textbf{Return} $Y_\rho$
\end{algorithmic}
\end{algorithm}
\begin{corollary}\label{cor:gm_lossless}
    With $M$ initialized as $M=\{(0, \infty), (\infty, f(\X))\}$, \Cref{alg:gaussian-multi-release} implements the Gaussian mechanism with lossless multiple release.
\end{corollary}
\begin{proof}
Let $\{ Y_{\rho} \}_{\rho \in S}$ be the set of outputs produced by \Cref{alg:gaussian-multi-release} on receiving the set $S$ of privacy parameters in arbitrary order.
Observe that the algorithm is implementing the (adaptive) sampling in \Cref{lemma:gaussian_seq_ooo}, and so it produces outputs with the same distribution as \Cref{lemma:gaussian_seq}.
    Property 1 of \Cref{def:lossless} follows immediately from observing that $Y_\rho \sim \NN(f(x), \frac{1}{2\rho})$.
    For property 2, note that every release in \Cref{lemma:gaussian_seq} can be viewed as randomized post-processing of the least private release.
    One way to see this is to note that for any two releases $Y_{\rho}$ and $Y_{\rho'}$ where $\rho < \rho'$, we have that $\mathrm{Cov}(Y_{\rho}, Y_{\rho'}) = \mathrm{Var}[Y_{\rho'}]$, implying that $Y_{\rho} = Y_{\rho'} + Z$ for aptly scaled zero-mean Gaussian noise $Z$.
\end{proof}

 \section{Extending to Independent Additive Noise}\label{sec:lossless-multiple-release}
We will next show that lossless multiple release holds for a larger class of mechanisms.
To proceed we introduce the notion of an \emph{independent additive noise mechanism}.
\begin{definition}[Independent Additive Noise Mechanism]\label{def:additive-noise}
We define an \emph{independent additive noise} \emph{mechanism} $\mathcal{A}_{f, \rho} : \XX \to \R^d$ as a mechanism of the form
\begin{equation*}
    \cA_{f, \rho}(\X) = f(\X) + Z\,,\quad Z \sim\mathcal{D}(\rho)\,,
\end{equation*}
where $\mathcal{D}(\rho)$ is a probability distribution parameterized in $\rho$, that draws a $d$-dimensional vector with i.i.d.\ samples.
\end{definition}

It turns out that \emph{any} \IAN mechanism $\mathcal{A}_{f,\rho}$ satisfying Definition~\ref{def:decomposable} supports lossless multiple release.
\begin{lemma}\label{lemma:additive_noise_means_lossless}
    Any \IAN mechanism $\mathcal{A}_{f,\rho}$ with noise distribution $\cD(\rho)$ satisfying convolution preorder can be implemented with lossless multiple release.
\end{lemma}
\begin{proof}
    The proof is for the one-dimensional case, but the same argument can be invoked for the multidimensional case as each coordinate has independent noise.
    We begin by proving that given $S=\{\rho_k : k\in[m]\}\subset \R^+$, where $\rho_1 < \dots < \rho_m$, it is possible to construct a set of releases $\{ Y_{\rho_k} : k\in[m]\}$ that satisfy \Cref{def:lossless}.
    By \Cref{def:decomposable}, it is possible to express each $Y_{\rho_k}$ as
    \begin{equation}\label{eq:multiple-release}
         Y_{\rho_k} = f(\X) + Z_{\rho_m} + \sum_{j=k}^{m-1} W_{\rho_{j+1}, \rho_j}\,,
    \end{equation}
    where $Z_{\rho_m}\sim\DD(\rho_m)$ and $W_{\rho_{j+1}, \rho_j}\sim\CC(\rho_{j+1}, \rho_j)$ are sampled independently.
    To prove that $Y_{\rho_k} \stackrel{d}{=} \mathcal{A}_{f, \rho_k}(\X)$, observe that $Y_{\rho_k}$ is distributed as $f(\X)$ plus a random variable drawn from $\DD(\rho_m) \ast \CC(\rho_m, \rho_{m-1}) \ast \dots \ast \CC(\rho_{k+1}, \rho_k) = \DD(\rho_k)$, as needed.
    For the second property in \Cref{def:lossless}, observe that conditioning on $Y_{\rho_m}=y$ we can express every other $Y_{\rho_k}$ for $k\in[m-1]$ as
    \begin{equation*}
        Y_{\rho_k} = y + \sum_{j=k}^{m-1} W_{\rho_{j+1}, \rho_j}\,.
    \end{equation*}
    As a result, the joint distribution on $(Y_{\rho_1}, \dots, Y_{\rho_m})_{| Y_{\rho_m}=y}$ is uniquely determined by $y$, proving the second property.

    We will now argue (inductively) that we can produce these releases adaptively.
    Let $S_1 \subset S_2 \dots \subset S_m = S$ be a sequence of subsets of $S$ where $|S_i| = i$.
    For the base case of $S_1$ we can make a single release using $\mathcal{A}_{f, \rho}$.
    For our inductive hypothesis, assume we have produced releases corresponding to all the privacy parameters in $S_n$.
    Now, for $S_{n+1}$, we re-label the privacy parameters such that $S_{n+1} = \{\rho_k\}_{k\in[n+1]}$ where $\rho_1 < \dots < \rho_{n+1}$.
    For the unique $\rho_i\in S_{n+1}\setminus S_n$, we can conditionally sample it from the joint distribution over $(Y_{\rho_k})_{k\in [n+1]}$, conditioned on the value of each release made in the previous round.
    By induction, it follows that we can \emph{adaptively} release $S$.
\end{proof}

\subsection{Sampling in Concrete Settings}
\Cref{lemma:additive_noise_means_lossless} proves the existence of a sampling procedure for lossless multiple release.
In \Cref{sec:gradual-release} we showed a sampling procedure in the particular case of Gaussian noise.
In \Cref{appendix:sampling} we show that the structure of \Cref{alg:gaussian-multi-release} holds for general \IAN{} mechanisms.
Namely, if the privacy parameters we support come from a bounded range $(\rho_0, \rho_\infty) \subset \mathbb{R}_+$, then the corresponding algorithm has a similar structure (see \Cref{alg:generic-sampling-simpler}).

More precisely, consider two neighboring releases $Y_{\rho_k}$ and $Y_{\rho_{k+1}}$ with noise parameters $\rho_k$ and $\rho_{k+1}$, respectively, and a new lossless release $Y_\rho$ with parameter ${\rho\in [\rho_{k}, \rho_{k+1}]}$.
We can use (\ref{eq:multiple-release}) on this set of $m+1$ releases and condition on the values of the $m$ previous releases $Y_{\rho_1},\dots,Y_{\rho_m}$.
Next, write $W_{\rho_{k+1},\rho_k} = W_1 + W_2$ where $W_1 = Y_{\rho}-Y_{\rho_{k+1}}$ and $W_2 = Y_{\rho_k}-Y_{\rho}$.
In~\Cref{appendix:sampling}, we show that sampling $Y_\rho$ can be reduced to the following task:
\begin{equation}\label{eq:how-to-sample}
    \text{Sample }\, W_1\, \text{ conditioned on }\ W_1 + W_2 = Y_{\rho_k} - Y_{\rho_{k+1}}.
\end{equation}
\paragraph{Example: Poisson mechanism}
Consider the \IAN mechanism using the Poisson distribution, $\text{Poi}(\lambda)$.
This mechanism has the property that noise is always a non-negative integer, making it a natural noise distribution for integer vectors in settings where negative noise is undesirable.
Appendix~\ref{sec:poisson} states some basic privacy properties of the Poisson mechanism.

The family of Poisson distributions parameterized by $\rho = 1/\lambda$ satisfies convolution preorder since for any $\lambda_1 > \lambda_2$, $\text{Poi}(\lambda_1) = \text{Poi}(\lambda_1 - \lambda_2) * \text{Poi}(\lambda_2)$.
To adaptively perform private lossless multiple release of a value $f(\X)$ using the Poisson mechanism and parameters ${\lambda_1 > \dots > \lambda_m}$ we notice that the ``bridging'' distribution in the $k$\textsuperscript{th} term of the sum in (\ref{eq:multiple-release}) has distribution $\text{Poi}(\lambda_k - \lambda_{k+1})$.
Note that without conditioning on $Y_{\rho_1},\dots,Y_{\rho_m}$, $W_1 \sim \text{Poi}(\lambda - \lambda_{k+1})$ and $W_2 \sim \text{Poi}(\lambda_k - \lambda)$.
By \Cref{lemma:conditional-poisson} in the appendix we have that the sampling in \eqref{eq:how-to-sample} reduces to $W_1 \sim \operatorname{Binomial}(N,p)$ for $N = Y_{\rho_k} - Y_{\rho_{k+1}}$ and $p = (\lambda - \lambda_{k+1})/(\lambda_k - \lambda_{k+1})$.

\paragraph{Example: Laplace Mechanism}
\citet{koufogiannis_2016} have already shown that the Laplace distribution parameterized by $\rho = 1/b$ satisfies convolution preorder for any scale parameters $b_1 > b_2$.
Let $\operatorname{LapBridge}(b_2, b_1)$ be the probability distribution that draws $0$ with probability $b_2^2/b_1^2$ and from $\lap(0, b_1)$ with the remaining probability.
Then $\lap(0, b_2) \ast \operatorname{LapBridge}(b_2, b_1) = \lap(0, b_1)$ exactly.
To implement lossless release via the sampling in \eqref{eq:how-to-sample}, it turns out that the conditional distribution is a mixture $W_1\sim\operatorname{LapBridge}(b, b_2)$ of three distributions:
With some probability $W_1$ is equal to either $0$ or $Y_{\rho_k}-Y_{\rho_{k+1}}$, and otherwise it is sampled from the convolution of two Laplace distributions. 
We also show that the related exponential distribution satisfies convolution preorder, see \cref{sec:laplace} for the full details.

\subsection{Lossless Multiple Release as a Blackbox}\label{sec:meta-theorem}
Inspired by Corollary 15 in \citet{koufogiannis_2016}, we provide a meta theorem for lossless multiple release based on \IAN{} mechanisms.
The central component is the following lemma, showing that the class of lossless multiple release mechanisms is closed under invertible post-processing.
\begin{lemma}\label{lemma:invertible_postp}
    Let $\MM_\rho : \XX \to \YY$ satisfy lossless multiple release, and let $H : \YY \to \YY'$ be an invertible function.
    Then $\MM_\rho' = H\circ \MM_\rho$ also satisfies lossless multiple release.
\end{lemma}
\begin{proof}
    Let $\MMM : \XX \times \R_+ \to \YY$ be an implementation of $\MM_\rho$ satisfying lossless multiple release, and let $\MMM' = H\circ \MMM$, which we will argue implements lossless multiple release.
    The only property that does not trivially hold for $\MMM'$ is the second property of \Cref{def:lossless}.
    For a set $S\subset \R_+$, note that to condition on $\MMM'(\X,\max(S))=y$ is equivalent to conditioning on ${H^{-1}(\MMM'(\X, \max(S))) = \MMM(\X, \max(S)) = H^{-1}(y)\in\YY}$.
    We thus have that conditioning $\MMM'(\X, \max(S))=y$ implies that the joint distribution over $(\MMM(\X, \rho))_{\rho\in S}$ is fully determined, and consequently so is ${(H(\MMM(\X, \rho)))_{\rho\in S} = (\MMM'(\X, \rho))_{\rho\in S}}$, and we are done. 
\end{proof}
\begin{theorem}[Meta theorem]\label{thm:meta}
    Let $\mathcal{M}_\rho : \mathcal{X} \to \mathcal{Y}$ be a family of mechanisms on a domain $\mathcal{X}$, indexed by a privacy parameter \(\rho\in \R_+\).
    Furthermore, assume $\mathcal{M}_\rho$ can be decomposed as $\mathcal{M}_\rho = H\circ\mathcal{A}_{f, \rho}$ where
    \begin{itemize}
        \item $\mathcal{A}_{f,\rho} : \XX \to \R^d$ is an \IAN mechanism for releasing $f : \XX \to \R^d$ with privacy parameter $\rho$ and with noise distribution satisfying a convolution preorder (\Cref{def:decomposable});
        \item $H : \R^d \to \mathcal{Y}$ is an \emph{invertible} post-processing step;
    \end{itemize}
Then $\mathcal{M}_\rho$ can be implemented with \emph{lossless multiple release}.\end{theorem}
\begin{proof}
    The theorem follows from invoking \Cref{lemma:additive_noise_means_lossless} together with \Cref{lemma:invertible_postp}.
\end{proof}

\paragraph{Supporting non-invertible post-processing}
To support non-invertible post-processing, we also introduce a weaker notion: \emph{weakly lossless multiple release}.
\vspace{2mm}
\begin{definition}[Weakly Lossless Multiple Release]\label{def:weakly-lossless}
   A mechanism  $\MM_\rho$ supports \emph{weakly lossless multiple release} if it can be written as  $\MM_\rho = H\circ\MM_{\rho}'$  where $\MM_\rho'$ supports lossless multiple release and $H$ is an arbitrary function (possibly chosen from some distribution).
\end{definition}
\vspace{2mm}
We define \emph{weakly lossless gradual release} analogously.
It follows that these classes of mechanisms are closed under all post-processing.
While this notion is indeed weaker, any algorithm $\MMM(\X, \rho)$ implementing weakly lossless multiple release for a $\rho$-private mechanism $\MM_\rho(\X)$, will have the property that the set of releases $(\MMM(\X, \rho))_{\rho\in S}$ are $\max(S)$-private, if $\rho$-privacy is closed under post-processing.
Examples of such privacy notions are $\varepsilon$-DP and $\rho$-zCDP.
We get the following immediate corollary to \Cref{thm:meta}.
\vspace{2mm}
\begin{corollary}\label{cor:noninv_postproc_meta}
    If the function $H$ in \Cref{thm:meta} is not invertible, then $\MM_\rho = H\circ A_{f,\rho}$ supports weakly lossless multiple release.
\end{corollary}
\vspace{2mm}
Weakly lossless multiple release will play a role in the next section, where we consider non-invertible post-processing such as truncation.

\vspace{2mm}
\section{Applications}\label{sec:applications}
In this section, we describe two applications supported by our framework for lossless multiple release: 
\begin{itemize}
\item Factorization mechanisms \cite{LiMHMR15}, where we want to privately release a linear query $A\X$, and
\item Sparse Gaussian histograms, also known as stability histograms \cite{wilkins2024,googlelibthreshold}.
\end{itemize}
We will make use of both \emph{lossless} multiple release (\Cref{def:lossless}), and its weaker variant (\Cref{def:weakly-lossless}).
This will be necessary since, e.g., the truncation used for sparse Gaussian histograms is not an invertible function, and so not covered by~\Cref{thm:meta}.
\vspace{2mm}

Because the noise generation for histograms on large domains is expensive, we give a dimension-independent algorithm that works in the gradual release setting.
Throughout, we state our results as post-processings of the Gaussian mechanism, but analogous results hold for any other \IAN{} mechanism meeting \Cref{def:decomposable}.
\pagebreak

\subsection{Lossless Multiple Release Factorization Mechanism}\label{sec:factorization-mechanism}

Let $A$ be a query matrix and fix a public factorization ${A = LR}$. Denote the \emph{factorization mechanism}~\cite{LiMHMR15} as ${\cF_\rho(\X) = A \X + L z = L (R \X + Z)}$ on some private dataset $\X$ where 
$Z$ is a vector drawn from a distribution $\cD(\rho)$ satisfying \cref{def:decomposable} parameterized by $\rho$, typically Gaussian or Laplace noise.
\begin{lemma}\label{lem:lossless-fact}
$\mathcal{F}_\rho$ can be implemented with weakly lossless multiple release, and lossless multiple release if $L$ has a left-inverse.
\end{lemma}
\begin{proof}
Observe that $\mathcal{F}_\rho(\X) = L \circ \mathcal{A}_{R, \rho}$ for an IAN mechanism $\mathcal{A}_{R, \rho}(\X) = R\X + Z$ where $Z\sim\cD(\rho)$ for a $\cD(\rho)$ satisfying \Cref{def:decomposable}.
The result follows from \Cref{thm:meta} and \Cref{cor:noninv_postproc_meta}, and noting that the linear map $L$ is invertible exactly when $L$ has a left-inverse.
\end{proof}

Besides proving existence in \cref{lem:lossless-fact}, we also give an explicit instantiation in \cref{alg:gradual-release-factorization} using the Gaussian mechanism.
\begin{algorithm}[t]
\caption{$\releaseFf$}\label{alg:gradual-release-factorization}
\begin{algorithmic}[1]
\PARAMETERS Factorization ${A = LR}$, $\ell_2$ sensitivity $\Delta_2$
\INPUTS Set of releases $M$, privacy parameter $\rho$
\STATE  Find ${(\rho_k, Y_{\rho_k}), (\rho_{k+1}, Y_{\rho_{k+1}})\in M}$ such that 
\STATE ~~~$\rho\in [\rho_k, \rho_{k+1})$ and $\forall (\rho', \cdot)\in M: \rho' \notin (\rho_k, \rho_{k+1})$\label{line:closest-rhos-fact}
\STATE $Z_\rho \sim  \NN\left(0, \Delta_2^2 \cdot \frac{(1-\rho_k/\rho)\cdot (1/\rho-1/\rho_{k+1})}{2(1 - \rho_k/\rho_{k+1})}\right)^d$\label{line:sample-new-fact}
\STATE Let $Y_\rho := L \cdot Z_\rho + \frac{(1 -\rho_k/\rho)Y_{\rho_{k+1}} + (\rho_k/\rho - \rho_k/\rho_{k+1}) Y_{\rho_k}}{1 - \rho_k / \rho_{k+1}}$ \label{line:estimate-new-fact}
\STATE Set $M = M \cup \{(\rho, Y_\rho)\}$\label{line:store-new-fact}
\STATE \textbf{Return} $Z$
\end{algorithmic}
\end{algorithm}
\begin{lemma}\label{lemma:gaussian_mechanism_is_lossless}
    With initialization $M=\{(0, \infty), (\infty, A\X)\}$, \Cref{alg:gradual-release-factorization} implements the Gaussian noise factorization mechanism $\mathcal{F}_\rho$ with (weakly) lossless multiple release.
\end{lemma}
\begin{proof}[Proof sketch.]
    Note that the algorithm is practically a copy of \Cref{alg:gaussian-multi-release}, but for the specific sensitive function $R\X$.
    The only structural difference is that the linearity of the post-processing $L$ allows for storing correlated Gaussian noise directly in $M$.
\end{proof}

\subsection{Weakly Lossless Multiple Release of Histograms}\label{sec:sparse-histograms}
We will now describe an example where the post-processing is neither linear nor invertible: releasing sparse (Gaussian) histograms \cite{korolova2009releasing,googlelibthreshold,BalleW18}.
Let $\X = (X_i)^n$ be a dataset of $n$ users, and we want to privately release the histogram $H(\X) = \sum_{i = 1}^n X_i$ where $X_i \in \{0,1\}^d$ and a user can contribute to $l$ distinct counts and the Gaussian mechanism which scales proportional to $\sqrt{l}$ is preferable.
In many natural settings, the domain size $d$ can be very large, and therefore, the resulting histogram is usually very sparse, $k = \|H(\X)\|_0  \ll d$.
Releasing a noisy histogram where noise has been added to \emph{each} coordinate would destroy the sparsity.
One can cope by introducing a threshold $\tau > 0$ where only noisy counts that exceed $\tau$ are released.
$\tau$ is usually set high enough such that the noise to a zero count will (after thresholding with high probability) still be zero.
We denote the \emph{support} of the histogram as ${U(H(\X)) = \{i \in [d]: H(\X)_i \neq 0\}}$ as the index of the non-zero coordinates.
Define the thresholding function ${T_{\tau} : \mathbb{R}^d \to \mathbb{R}^d}$ as:
\[T_{\tau}(x) = \left( y_i \right)_{i=1}^{d}, \quad \text{where} \quad {y_i =
\begin{cases} 
x_i, & \text{if } x_i \geq \tau \\
0, & \text{otherwise}
\end{cases}}\,.
\]

Denote the (\IAN{}) sparse histogram mechanism as ${\mathcal{H}_{\rho}(\X)= T_\tau(H(\X) + Z})$ where ${Z \sim \mathcal{D}(\rho)}$ is a distribution satisfying a convolution preorder (\cref{def:decomposable}).
\begin{lemma}\label{lem:meta-gsh}
$\mathcal{H}_\rho$ can be implemented with weakly lossless multiple release.
\end{lemma}
\begin{proof}
Observe that $\mathcal{H}_\rho$ can be expressed as a composite function $T_\tau\circ \mathcal{A}_{\rho,f}(\X)$ for an \IAN mechanism $\mathcal{A}_{\rho,f}(\X)$ meeting \Cref{def:decomposable}.
The statement thus follows from \Cref{cor:noninv_postproc_meta}.
\end{proof}
It is straightforward to implement weakly lossless gradual release for stability histograms using our framework with time and space complexity linear in dimension $d$; see \Cref{alg:gradual-release-histogram}.
The following lemma is proved in~\Cref{sec:efficient-sparse}.
\begin{lemma}\label{lemma:gradual_histograms}
    With initialization $M=\{(0, 0)\}$, \Cref{alg:gradual-release-histogram} implements $\mathcal{H}_\rho$ with weakly lossless gradual release.
\end{lemma}
\begin{algorithm}[t]
\caption{\GLSGH}
\label{alg:gradual-release-histogram}
\begin{algorithmic}[1]
\PARAMETERS $\ell_2$ sensitivity $\Delta_2$ \\
\INPUTS histogram $H(\X)$, set of releases $M$,\\
privacy parameter $\rho$, threshold $\tau$\\
\STATE Extract the single element ${(\rho', Z')\in M}$ 
\STATE Sample $\tilde{Z} \sim \NN\left(0, \frac{1}{2}\Delta_2^2(\rho - \rho')\right)^d$
\STATE Let $Z  = \dfrac{\rho'}{\rho}Z' + \dfrac{1}{\rho}\tilde{Z}$
\FOR{each $i\in [d]$}\label{gsh-post-processing}
    \STATE{\textbf{if} $H(\X)_i + Z_i > \tau$ \textbf{then}} $Y_i = H(\X)_i + Z_i$
    \STATE{\textbf{else}} $Y_i = 0$
\ENDFOR
\STATE Set $M = \{(\rho, Z)\}$\label{line:store-new}
\STATE \textbf{Return} $Y$
\end{algorithmic}
\end{algorithm}

\paragraph{Improving efficiency for weakly lossless gradual release}
\citet*{korolova2009releasing} showed that under approximate differential privacy, one can skip the noise generation for zero counts if thresholding is enabled, adding only a small probability for infinite privacy loss if this coordinate is non-zero in a neighboring dataset. 
Unfortunately, this approach does not fit our meta theorem, so instead, we apply a technique due to \citet*{cormode2012dpsummaries} that efficiently simulates the noise distribution of those zero counts that would exceed the threshold and thus gives an identical distribution to the truncated noisy histogram.
The procedure is described in the following lemma.
\begin{lemma}\label{lemma:smart-histogram-sampling}
    For a fixed noise distribution $\mathcal{D}$, the output distribution of releasing $\cT_\tau(H(\X) + Z)$ with $Z\sim \mathcal{D}$ are independent noise samples is equal to the following process:
    \textbf{(a)} add noise to every non-zero coordinate in ${H(\X)_i}$ for ${i \in U(H(\X))}$ and then \textbf{(b)} draw ${q\sim\operatorname{Binomial}(d-k, p)}$ where ${p = \Pr_{Z \sim \mathcal{D}}[Z > \tau]}$.
    \textbf{(c)} Sample a subset ${Q\subseteq [d]\setminus U(H(\X))}$ of size $|Q| = q$ uniformly at random and \textbf{(d)} set the noise for ${H(\X)}$ to ${Z \sim \mathcal{D}}$ conditioned on being above the threshold $\tau$.
\end{lemma}
\begin{proof}[Proof sketch]
A formal argument can be found in~\citet{cormode2012dpsummaries}, here we just provide a sketch. 
It is clear that every entry in the support gets the correct output distribution.
Furthermore, the number of zero counts $q$ whose noise exceeds $\tau$ follows a binomial distribution, and we can add conditional noise to a uniformly drawn subset of size $q$.
\end{proof}

\Cref{alg:efficient-gradual-release} in~\Cref{sec:efficient-sparse} implements the efficient routine described by \Cref{lemma:smart-histogram-sampling} under weakly lossless gradual release.
The key challenge is that the probability for a histogram count to exceed the threshold in a given round now depends on whether it exceeded the threshold in any prior round.
Intuitively, all zero counts that have never exceeded the threshold have an equal probability to exceed the threshold in any given round, and so the simulation technique in \Cref{lemma:smart-histogram-sampling} can be used for these counts (with different probabilities and conditional distributions).
However, once any zero count exceeds the threshold, it will have to be treated the same as non-zero counts in all future rounds.

The utility and privacy of \Cref{alg:efficient-gradual-release} is proved by arguing that its output distribution matches that of \Cref{alg:gradual-release-histogram}.
A proof sketch for the following lemma is given in \Cref{sec:efficient-sparse}.
\begin{lemma}\label{lemma:efficient-histogram-id}
    Let $H(\mathbf{x})$ be a histogram over $[d]$, $\rho_1, \dots, \rho_m$ be a sequence increasing of privacy budgets, $\tau_1, \dots, \tau_m$ be a sequence of thresholds and $\Delta_2 > 0$ the $\ell_2$-sensitivity of the histogram.
    Also let the sequences of outputs $(Y^{(1)}, \dots, Y^{(m)})$ and $(\hat{Y}^{(1)}, \dots, \hat{Y}^{(m)})$ be derived from running \Cref{alg:gradual-release-histogram} and \Cref{alg:efficient-gradual-release} respectively with the preceding parameters as input.
    Then the sequences $(Y^{(1)}, \dots, Y^{(m)})$ and $(\hat{Y}^{(1)}, \dots, \hat{Y}^{(m)})$ are identically distributed.
\end{lemma}
The benefit of the efficient sampling technique is that the number of sampled Gaussians will \emph{never} exceed that of the na\"ive approach and can potentially be in the order of the sparsity depending on the parameter regime.
\begin{lemma}\label{lemma:sampled_gaussians}
   Let $c$ be the number of zero-counts in the histogram output by \Cref{alg:efficient-gradual-release} least once over its execution.
   Then \Cref{alg:efficient-gradual-release} will sample $(k+c)m$ (truncated or otherwise) Gaussians.
\end{lemma}
We again refer the reader to \Cref{sec:efficient-sparse} for a proof.

\section{Empirical Evaluation}\label{sec:empirical-evaluation}

\begin{figure}[t]
  \centering
    \includegraphics[width=1.0\linewidth]{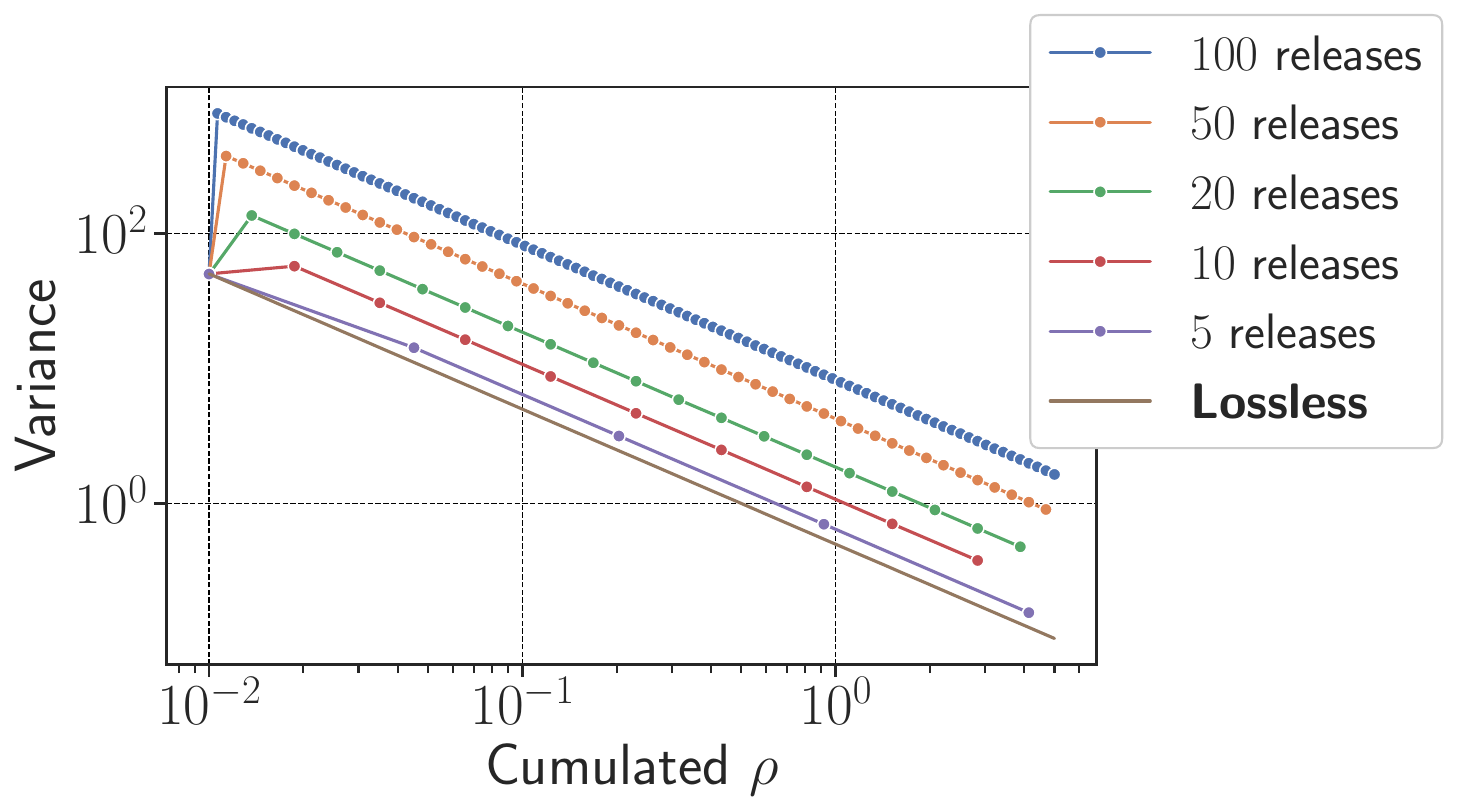}
  \caption{Accuracy comparison of multiple uncorrelated releases compared to lossless multiple release. 
  Budgets are spaced evenly on a logarithmic scale between $\rho=0.001$ and $\rho=5$ on the x-axis. 
  Creating independent releases with a denser set of privacy parameters comes at the cost of increased variance.
  In the lossless setting we get the best possible variance with no bound on the number of releases.
  } \label{fig:experiments}
\end{figure}

To confirm our theoretical claims, we empirically evaluate the accuracy of lossless multiple release against a baseline algorithm that uses independent releases.
We evaluate the impact by focusing on noise in isolation to avoid capturing the effect of specific queries. 
The baseline algorithm is a simple Gaussian mechanism where noise is drawn independently for each consecutive release. 
To showcase our algorithm's performance, we demonstrate how the cost incurred by uncoordinated releases grow with the amount of releases, in contrast to the lossless multiple release where there is no additional cost. 
We repeat our experiments $10^6$ times, and measure the variance of the noise. 
The plot (\Cref{fig:experiments}) shows, as expected, that our mechanism does not lose any utility from making multiple releases. 
As we can see, there is an expected increase in variance when going from one release to multiple releases---the initial jump is larger the more releases we want to make as the budget used for the second release is the difference between the starting point ($\rho=0.001$) and a constant increase in budget, whereas the subsequent releases all have the same difference in budget.

\section{Conclusion and Open Questions}\label{sec:conclusion}

We have initiated a systematic study of differential privacy with multiple releases, motivated by settings in which many levels of privacy or trust may co-exist.
The main message is that it is possible to generalize a large class of lossless gradual release techniques to this setting, even when releases are determined online and in no particular order.
For mechanisms based on Gaussian, Laplace, or Poisson additive noise we give simple and efficient sampling procedures for creating new lossless releases.
In particular, we are able to do lossless multiple release for any factorization mechanism using invertible matrices.
Finally, we consider algorithmic challenges related to lossless gradual release and show that private sparse histograms may be computed much more efficiently than what a direct application of our general results would imply.

There are still many open questions concerning mechanisms that do not inherit their privacy guarantee from an \IAN mechanism.
In particular, it would be interesting to determine under which conditions the exponential mechanism~\citep{mcsherry2007mechanism} supports lossless multiple release. 
\citet{koufogiannis_2016} conjecture that this is always possible.
Other central private algorithms, such as report noisy max~\cite{dwork2014algorithmic} also do not have known lossless multiple release mechanisms, even in the gradual release setting.
A final challenge we would like to mention is implementing our mechanisms on a finite computer, e.g., creating a multiple release version of the discrete Gaussian mechanism~\citep{canonne2020discrete}.
An appealing approach would be to base this on Poisson noise, which meets the technical conditions for our framework and approaches the Gaussian distribution in the limit.

\section*{Acknowledgments}
Andersson, Pagh and Retschmeier carried out this work at Basic Algorithms Research Copenhagen (BARC), which was supported by the VILLUM Investigator grant 54451. 
Providentia, a Data Science Distinguished Investigator grant from the Novo Nordisk Fonden, supported Andersson, Pagh and Retschmeier. 
Nelson carried out part of this work at Uppsala University.
 
\section*{Impact Statement}
This paper presents work whose goal is to advance the field of private machine learning. There are many potential societal consequences of our work, none which we feel must be specifically highlighted here.

\medskip

\bibliographystyle{icml2025}
\bibliography{main} 

\appendix
\onecolumn

\section{On Lossless Multiple Release Sampling}\label{appendix:sampling}
\Cref{lemma:additive_noise_means_lossless} makes no claim about how easy it is to sample noise from the conditional noise distribution, but guarantees its existence.
We will use this section for discussing this sampling in greater detail.
Consider the joint distribution of the releases defined by \cref{eq:multiple-release} from the proof of \Cref{lemma:additive_noise_means_lossless}, re-stated below for convenience:
\begin{equation}\tag{1}
     Y_{\rho_k} = f(\X) + Z_{\rho_m} + \sum_{j=k}^{m-1} W_{\rho_{j+1}, \rho_j}\,,\qquad k=1, \dots, m\,,
\end{equation}
where $Y_{\rho_k}$ is the $\rho_k$-private release, $Z_{\rho_m}\sim \mathcal{D}(\rho_m)$ and $W_{\rho_{j+1}, \rho_j}\sim\mathcal{C}(\rho_{j+1}, \rho_j)$.
Recall that $\mathcal{C}(\rho', \rho)$ is the unique distribution where for any $0 < \rho < \rho' : \mathcal{D}(\rho) = \mathcal{D}(\rho') * \mathcal{C}(\rho', \rho)$.
As argued for in the proof of \Cref{lemma:additive_noise_means_lossless}, to make a new $\rho$-private release, we consider the joint distribution, and sample the new releases conditioned on all past releases.

Formally, for a set of privacy parameters $\rho_1 < \dots < \rho_m$, we are interested in releasing $\rho_k$ for a single $k\in[m]$, conditioned on the set of past releases $M=\{(\rho_i, Y_{\rho_i}) : i\in[m]\setminus\{k\}\}$.
We give pseudocode for this in \Cref{alg:generic-sampling}, and a lemma for its correctness.
\begin{algorithm}[t]
\caption{$\releaseGeneric$}\label{alg:generic-sampling}
\begin{algorithmic}[1]
\PARAMETERS Noise distribution family $\mathcal{D}(\rho)$ satisfying a convolution preorder,\\
bridging noise distribution $\mathcal{C}(\rho, \rho')$ where $\mathcal{D}(\rho') * \mathcal{C}(\rho', \rho) = \mathcal{D}(\rho)$ for $0 < \rho < \rho'$.
\INPUTS Sensitive value $f(\X)$, new privacy parameter $\rho_k$, set of past releases $M = \{ (\rho_i, Y_{\rho_i}) : i\in[m]\setminus\{ k \} \}$
\IF{$M = \emptyset$}
\STATE Sample $Z_{\rho_k}$ from $\mathcal{D}(\rho_{\rho_k})$
\STATE Set $Y_{\rho_k} = f(\X) + Z_{\rho_k}$
\ELSIF{$k = 1$}
\STATE Sample $W_{\rho_{k+1}, \rho_{k}}$ from $\mathcal{C}(\rho_{k+1}, \rho_{k})$
\STATE Set $Y_{\rho_k} = Y_{\rho_{k+1}} + W_{\rho_{k+1}, \rho_{k}}$
\ELSIF{$k\in (1, m)$} \label{case3.sample}
\STATE Sample $W_{\rho_{k+1}, \rho_k}$ from $\mathcal{C}(\rho_{k+1}, \rho_k)$ conditioned on $W_{\rho_{k+1}, \rho_k} + W_{\rho_k, \rho_{k-1}} = Y_{\rho_{k-1}} - Y_{\rho_{k+1}}$ 
\STATE Set $Y_{\rho_k} = Y_{\rho_{k+1}} + W_{\rho_{k+1}, \rho_k}$
\ELSIF{$k=m$}
\STATE Sample $Z_{\rho_k}$ from $\mathcal{D}(\rho_k)$ conditioned on $Z_{\rho_k} + W_{\rho_{k}, \rho_{k-1}} = Y_{\rho_{k-1}} - f(\X)$
\STATE Set $Y_{\rho_k} = f(\X) + Z_{\rho_k}$
\ENDIF
\STATE Set $M = M \cup \{(\rho_k, Y_{\rho_k})\}$
\STATE \textbf{Return} $Y_{\rho_k}$
\end{algorithmic}
\end{algorithm}
\begin{lemma}\label{lemma:generic_sampling}
    Initializing $M=\emptyset$ and then running \Cref{alg:generic-sampling} for a set of privacy parameters $\{\rho_i : i\in[m]\}$ (processed in arbitrary order), will produce a set of outputs $\{ Y_{\rho_i} : i\in[m]\}$ with distribution given by \cref{eq:multiple-release}.
\end{lemma}

\begin{proof}
    For $M=\emptyset$, we have that $Y_{\rho_k} = f(\X) + \cD(\rho_k) =\mathcal{A}_{f,\rho_k}(\X)$ on lines 2-3, as expected, where $\cA$ is our \IAN noise mechanism releasing $f$ with $\rho$-privacy.
    For the remaining cases, assume $M$ contains $m-1$ releases with the correct joint distribution, and we are generating a new release at privacy level $\rho_k$.
    We will argue that \Cref{alg:generic-sampling} produces a $Y_{\rho_k}$ whose joint distribution with the releases in $M$ matches \eqref{eq:multiple-release}.
    For $k=1$, observe that $Y_{\rho_1} = Y_{\rho_2} + W_{\rho_{2}, \rho_{1}}$, and so lines 5-6 are correct.
    For $1 < k < m$, note that $Y_{\rho_k} = Y_{\rho_{k+1}} + W_{\rho_{k+1}, \rho_{k}}$ and ${Y_{\rho_{k-1}} = Y_{\rho_k} + W_{\rho_{k}, \rho_{k-1}}}$, which combined allows us to identify the correct conditional distribution.
    $Y_{\rho_k}$ is given by $Y_{\rho_{k+1}} + W_{\rho_{k+1}, \rho_{k}}$, conditioned on $W_{\rho_{k+1}, \rho_{k}} + W_{\rho_{k}, \rho_{k-1}} = Y_{\rho_{k-1}} - Y_{\rho_{k+1}}$, matching lines 8-9.
    For the final case of $k=m$, then $Y_{\rho_m} = f(\X) + Z_{\rho_m}$ and we have that $Y_{\rho_{m-1}} = Y_{\rho_m} + W_{\rho_{m}, \rho_{m-1}}$.
    It follows that the correct noise distribution is $Y_{\rho_m} = f(\X) + Z_{\rho_{m}}$, conditioned on $Z_{\rho_m} + W_{\rho_{m}, \rho_{m-1}} = Y_{\rho_{m-1}} - f(\X)$, matching lines 11-12.
    An inductive argument identical to that given in the proof of \Cref{lemma:additive_noise_means_lossless} completes the proof.
\end{proof}

\subsection{Simplification and Removing Dependency on the Dataset}
Note that $f(\X)$ is only used in \Cref{alg:generic-sampling} when a more accurate release is generated (lines 1 and 10).
In the remaining cases we are only adding noise to past releases.
This already allows us to simplify the algorithm, and argue for not having to store $f(\X)$ in memory indefinitely.
The algorithm reduces to the case (see line~7 of \cref{alg:generic-sampling}) where only the two closest releases are combined into a new one.
Such an algorithm is given in \Cref{alg:generic-sampling-simpler}, together with \Cref{lemma:generic_sampling_simplified}.

\begin{algorithm}[t]
\caption{$\releaseGenericS$}
\label{alg:generic-sampling-simpler}
\begin{algorithmic}[1]
\PARAMETERS Noise distribution family $\mathcal{D}(\rho)$ satisfying a convolution preorder,\\
bridging noise distribution $\mathcal{C}(\rho, \rho')$ where $\mathcal{D}(\rho') * \mathcal{C}(\rho', \rho) = \mathcal{D}(\rho)$ for $0 < \rho < \rho'$
\INPUTS Privacy parameter $\rho$, set of releases $M$
\STATE Find ${(\rho_k, Y_{\rho_k}), (\rho_{k+1}, Y_{\rho_{k+1}})\in M}$ such that $\rho\in [\rho_k, \rho_{k+1})$ and $\forall (\rho', \cdot)\in M: \rho' \notin (\rho_k, \rho_{k+1})$\label{line:closest-rhos}
\STATE Sample $W_{\rho_{k+1}, \rho}$ from $\mathcal{C}(\rho_{k+1}, \rho)$ conditioned on $W_{\rho_{k+1}, \rho} + W_{\rho, \rho_{k}} = Y_{\rho_{k}} - Y_{\rho_{k+1}}$
\STATE Set $Y_\rho = Y_{\rho_{k+1}} + W_{\rho_{k+1}, \rho}$
\STATE Set $M = M \cup \{(\rho, Y_\rho)\}$
\STATE \textbf{Return} $Y_\rho$
\end{algorithmic}
\end{algorithm}

\begin{lemma}\label{lemma:generic_sampling_simplified}
     Let $M=\{(\rho_0, Y_{\rho_0}), (\rho_\infty, Y_{\rho_\infty})\}$ for $Y_{\rho_\infty} = \mathcal{A}_{\rho_\infty, f}(x)$, and $Y_{\rho_0} = Y_{\rho_\infty} + W_{\rho_\infty, \rho_0}$.
     Then running \Cref{alg:generic-sampling-simpler} with input $M$ and a set of privacy parameters $\{\rho_i : i\in[m]\}\subset (\rho_0, \rho_\infty)$ (processed in arbitrary order), will produce a set of outputs $\{Y_{\rho_i} : i\in[m]\}$ with distribution given by \cref{eq:multiple-release}.
     Moreover, the set $M$ will at all times satisfy $\rho_\infty$-privacy.
\end{lemma}
\begin{proof}
    The statement follows from carefully comparing to \Cref{alg:generic-sampling}.
    Let $S=(\rho_1, \dots, \rho_m)$ be a sequence of noise values from the lemma statement, and let $O=(Y_{\rho_1}, \dots, Y_{\rho_m})$ be the outputs.
    Consider a second sequence $S'=(\rho_\infty, \rho_0, \rho_1, \dots, \rho_m)$, and consider the corresponding sequence of outputs $O'=(Y_{\rho_\infty}', Y_{\rho_0}', Y_{\rho_1}', \dots, Y_{\rho_m}')$ from inputting $S'$ and $M'=\emptyset$ to \Cref{alg:generic-sampling}.
    We can directly check that $Y_{\rho_0}\stackrel{d}{=}Y_{\rho_0}'$ and $Y_{\rho_\infty}\stackrel{d}{=} Y_{\rho_\infty}'$.
    For the remaining outputs, note that just before each algorithm is called, their internal states $M$ and $M'$ are also identically distributed.
    Now, the fact that each $\rho_i\in(\rho_0, \rho_\infty)$ implies that the remaining outputs $Y_{\rho_1}', \dots, Y_{\rho_m}'$ are generated from lines 10-12 in \Cref{alg:generic-sampling}.
    Checking carefully, lines 1-3 in \Cref{alg:generic-sampling-simpler} are implementing the same routine, and since $M \stackrel{d}{=} M'$, it follows that $(Y_{\rho_1}, \dots, Y_{\rho_m}) \stackrel{d}{=} (Y_{\rho_1}', \dots, Y_{\rho_m}')$, and so the statement follows from \Cref{lemma:generic_sampling}.
    The last statement on the $\rho_\infty$-privacy of $M$ follows from \Cref{alg:generic-sampling-simpler} implements \Cref{eq:multiple-release}, and so each release in $M$ can at any time during execution be viewed as randomized post-processing of $Y_\infty$.
\end{proof}
Essentially, if we commit to supporting a bounded range of privacy parameters then we get a simpler algorithm.
\Cref{lemma:generic_sampling_simplified} also says something more:
If we commit to supporting a lowest level of privacy $\rho_\infty$, then \Cref{alg:generic-sampling-simpler} can be implemented in such a way that its internal state is $\rho_\infty$-private.
After initializing $M$ using the sensitive function $f(\X)$, we can erase $f(\X)$ from memory and $Y_\infty$ will contain enough private information to generate all future releases.
This could prove useful in settings where the time between releases is large, and we want to limit the private information leaked if the state of the algorithm were to be compromised.
This can be compared with the gradual release setting, where natural implementations would require consistent access to $f(\X)$.

\section{Omitted Proof for Gaussian Lossless Multiple Release}\label{appendix:omitted-proofs}

\begin{proof}[Proof of \Cref{lemma:gaussian_seq_ooo}]
    Throughout the proof we will ignore the factor $2$ in the denominator of the variance.
To argue that the values generated by the process match the distribution \Cref{lemma:gaussian_seq}, we will argue for increasing subsets of releases.
    Let $S_n = \{ \rho_k : k\in [n-1]\}$ where $0 < \rho_1 < \dots < \rho_n < \rho_\infty$ be the set of values in $S$ for which we have generated Gaussians at the start of the $n$\textsuperscript{th} round of the process.
    Note that we re-label the $\rho$'s between the rounds such that $\rho_k\in S_n$ always is the $k$\textsuperscript{th} smallest value in $S_n$.
    Our argument will proceed by induction: assume that all the Gaussians $\{Y_{\rho_k} : k\in [n]\}$ generated after $n$ rounds have the distribution given by \Cref{lemma:gaussian_seq}.
    Then we will show that $\{Y_{\rho_k} : k\in[n-1]\}\cup \{ Y_\rho \}$ has the same distribution as predicted by $S_{n}\cup\{\rho\}$ from invoking \Cref{lemma:gaussian_seq}.
    
    We begin with our base case.
    For $n=1$, we have that $\rho_l = 0$ and $\rho_r=\rho_\infty$, and so $Y_\rho = Y_{\rho_\infty} + \NN(0, 1/\rho - 1/\rho_\infty)$.
    Since $Y_{\rho_\infty}\sim\NN(\beta, 1/\rho_\infty)$, we have that $Y_\rho\sim\NN(\beta, 1/\rho)$, as expected, and so the base case passes.

    For $n\geq 2$, we first consider the following cases.
    
    \emph{Case 1}: $\rho_l = 0, \rho_r=\rho_1\in S_n$.
    In this case, $Y_\rho = Y_{\rho_1} + \NN(1/\rho - 1/\rho_1)$, and so $Y_\rho\sim\NN(\beta, 1/\rho)$.
    Since $\rho < \rho_1$, we have that $\forall i\in[n-1] : \mathrm{Cov}(Y_\rho, Y_{\rho_i}) = \mathrm{Cov}(Y_{\rho_1}, Y_{\rho_i}) = 1/\rho_i$, as expected for the release with the smallest value in $\{\rho\}\cup S_n$, and so the case is complete.

    \emph{Case 2}: $\rho_l = \rho_{n-1}, \rho_r = \rho_\infty = \infty$.
    We have to deal with the case where we have set $\rho_\infty=\infty$ separately, as it is a bit of a trick.
    Note that we get $Y_\rho = \frac{(1-\rho_{n-1})\beta + \rho_{n-1}Y_{\rho_{n-1}}}{\rho}  + \NN(0, \frac{1-\rho_{n-1}/\rho}{\rho})$, since $Y_{\rho_\infty} = \beta$ in this case.
    The end result is once more a sum of Gaussians, with mean $\beta$, and for the variance we can explicitly compute that $\var[Y_\rho] = \frac{\rho_{n-1}^2}{\rho_{n-1}\rho^2} - \frac{\rho - \rho_{n-1}}{\rho^2} = 1/\rho$, as expected.
    For the covariance, for $i\in[n-1] : \mathrm{Cov}(Y_\rho, Y_{\rho_i}) = \frac{\rho_{n-1}}{\rho}\mathrm{Cov}(Y_{\rho_{n-1}}, Y_{\rho_i}) = 1/\rho$, as expected for the largest value in $\{\rho\}\cup S_n$, and so we are done.
    Now we consider the most general case.
    
    \emph{Case 3}: $\rho_l\in S_{n-1}$ and $\rho_r \neq \infty$.
We begin with computing the variance of $Y$ using the hypothesis $\var[Y_{\rho_r}] = 1/\rho_r$ and $\var[Y_{\rho_l}] = 1/\rho_l$
    \begin{align*}
        \var[Y_\rho] &=  \frac{(\rho_l^{-1} -\rho^{-1})(\rho^{-1} - \rho_{r}^{-1})}{\rho_l^{-1} - \rho_{r}^{-1}}
        + \frac{(\rho_l^{-1} - \rho^{-1})^2\var[Y_{\rho_r}] + (\rho^{-1} -\rho_{r}^{-1})^2\var[Y_{\rho_l}]}{(\rho_l^{-1} - \rho_{r}^{-1})^2}\\
        &= \frac{1}{(\rho_l^{-1} - \rho_{r}^{-1})^2}\bigg[(\rho_l^{-1} - \rho^{-1})(\rho^{-1} - \rho_{r}^{-1})(\rho_l^{-1} - \rho_{r}^{-1}) +\rho_{r}^{-1}(\rho_l^{-1} - \rho^{-1})^2 + \rho_l^{-1}(\rho^{-1} - \rho_{r}^{-1})^2 \bigg]\\
        &= \frac{1}{(\rho_l^{-1} - \rho_{r}^{-1})^2}\cdot\rho^{-1}(\rho_l^{-1} - \rho_{r}^{-1})^2 = \frac{1}{\rho}\,,
    \end{align*}
    where the third equality follows from applying the identity
    \begin{align*}
        (a-b)(b-c)(a-c) + c(a-b)^2 + a(b-c)^2 = b(a-c)^2\,,
    \end{align*}
    to the expression in the brackets for $a = \rho_l^{-1}, b = \rho^{-1}$ and $c = \rho_{r}^{-1}$, 
    proving the correctness of the variance.
    Furthermore, $Y_\rho$ is a sum of Gaussians and a convex combination of two Gaussians with mean $\beta$, so $Y_\rho\sim\NN(\beta, 1/\rho)$.
    What remains to show is that the covariances match up, which we do next.

    We start with the case where $\rho_r \neq \rho_\infty$, and so there exists $\rho_k, \rho_{k+1}\in S_n$ such that $\rho\in (\rho_k, \rho_k+1)$
    For $j\in[n-1]$, we therefore have that
    \begin{align*}
        \mathrm{Cov}(Y_\rho, Y_{\rho_j}) &=
        \mathrm{Cov}\bigg(\frac{(\rho_l^{-1} - \rho^{-1})Y_{\rho_r} + (\rho^{-1} - \rho_{r}^{-1})Y_{\rho_l}}{\rho_l^{-1} - \rho_{r}^{-1}}, Y_{\rho_j}\bigg)\\
        &= \frac{1}{\rho_l^{-1} - \rho_{r}^{-1}}\bigg[(\rho_l^{-1} - \rho^{-1})\mathrm{Cov}(Y_{\rho_r}, Y_{\rho_j}) + (\rho^{-1} - \rho_{r}^{-1})\mathrm{Cov}(Y_{\rho_l}, Y_{\rho_j})\bigg]\\
        &= \frac{1}{\rho_l^{-1} - \rho_{r}^{-1}}\bigg[(\rho_l^{-1} - \rho^{-1})\min(\rho_{r}^{-1}, \rho_j^{-1}) + (\rho^{-1} - \rho_{r}^{-1})\min(\rho_{l}^{-1}, \rho_j^{-1})\bigg]\,.
\end{align*}
    We consider the cases of $j\leq k$ and $j > k$ separately.
    For $j\leq k$, we have that
    \begin{equation*}
        \mathrm{Cov}(Y_\rho, Y_{\rho_j}) = \frac{1}{\rho_l^{-1} - \rho_{r}^{-1}}\bigg[(\rho_l^{-1} - \rho^{-1})\rho_{r}^{-1} + (\rho^{-1} - \rho_{r}^{-1})\rho_{l}^{-1}\bigg]\enspace = \frac{1}{\rho}\,,
    \end{equation*}
    and similarly for $j \geq k+1$ we get
    \begin{equation*}
        \mathrm{Cov}(Y_\rho, Y_{\rho_j}) = \frac{1}{\rho_k^{-1} - \rho_{k+1}^{-1}}\bigg[(\rho_k^{-1} - \rho^{-1})\rho_{j}^{-1} + (\rho^{-1} - \rho_{k+1}^{-1})\rho_{j}^{-1}\bigg]\enspace = \frac{1}{\rho_j}\,.
    \end{equation*}
    It follows that $\mathrm{Cov}(Y_\rho, Y_{\rho_j}) = \frac{\Delta_2^2}{2\max(\rho, \rho_j)}$ for $j\in[m]$, and so we are done with this part of the covariances.

    For the last step, we consider what happens when $\rho_r=\rho_\infty < \infty$.
    In this case, $\rho_l = \rho_{n-1}$ and so for $j\in[n-1]$:
    \begin{align*}
        \mathrm{Cov}(Y_\rho, Y_{\rho_j}) &= \mathrm{Cov}\left(\frac{(1-\rho_{n-1}/\rho)Y_{\rho_\infty} + (\rho_{n-1}/\rho - \rho_{n-1}/\rho_\infty)Y_{\rho_{n-1}}}{1-\rho_{n-1}/\rho_\infty}, Y_{\rho_j}\right)\\
        &= \frac{1-\rho_{n-1}/\rho}{1-\rho_{n-1}/\rho_\infty}\mathrm{Cov}(Y_{\rho_\infty}, Y_{\rho_j}) + \frac{\rho_{n-1}/\rho - \rho_{n-1}/\rho_\infty}{1-\rho_{n-1}/\rho_\infty}\mathrm{Cov}(Y_{\rho_{n-1}}, Y_{\rho_j})\\
        &= \frac{1-\rho_{n-1}/\rho}{\rho_\infty-\rho_{n-1}} + \frac{1/\rho - 1/\rho_\infty}{1-\rho_{n-1}/\rho_\infty} = \frac{1 - \rho_{n-1}/\rho + \rho_\infty/\rho - 1}{\rho_\infty - \rho_{n-1}} = 1/\rho\,,
    \end{align*}
    and now we are done.
\end{proof}

\section{Details on (Efficient) Weakly Lossless Gradual Release of Sparse Gaussian Histograms}\label{sec:efficient-sparse}
\Cref{alg:gradual-release-histogram} implements \emph{weakly lossless gradual release} of private histograms.
We prove this next.
\begin{proof}[Proof of \Cref{lemma:gradual_histograms}]
    Observe that for any increasing sequence $(\rho_k)_{k\in m}$, the corresponding sequence of variables $Z$ on line~3 produced by the algorithm, have the same distribution as in \Cref{lemma:inverse-variance-weighting-is-nice} for $\beta=0$.
    The $Y$ that is ultimately returned, however, is a post-processing of $H(\X) + Z$, which has the same distribution as \Cref{lemma:inverse-variance-weighting-is-nice} for $\beta=H(\X)$.
    Therefore \Cref{alg:gradual-release-histogram} is implementing lossless gradual release for the Gaussian mechanism applied to $H(\X)$, combined with a non-invertible post-processing.
    The lemma statement follows from \Cref{cor:noninv_postproc_meta}.
\end{proof}
Note that one \emph{only} has to store the noisy terms from the preceding round to implement \cref{alg:gradual-release-histogram}.
Nevertheless, it might be infeasible, say, when the domain is really large, to sample the noise for the zero coordinates.
\Cref{alg:efficient-gradual-release} uses the computational trick in \Cref{lemma:smart-histogram-sampling} to speed up this computation.
The algorithm is static, where the privacy budgets are fixed upfront but can easily be converted to an online algorithm.
Recall that $U(H(\X)) = \{i \in [d]: H(\X)_i \neq 0\}$ is the support of the histogram.
\begin{algorithm}[t]
\caption{\EGLSGH}
\label{alg:efficient-gradual-release}
\begin{algorithmic}[1]
\PARAMETERS $\ell_2$-sensitivity $\Delta_2$
\INPUTS Private histogram $H(\X)$,
privacy budgets $\rho_1 < \dots < \rho_m$,
thresholds $\tau_1,\dots,\tau_m$
\STATE Let $S^{(0)} = U(H(\X))$ {\bf \color{blueish} // Tracking counts that have already been released}
\STATE Also let $Z^{(0)} = \{0\}^d$ and $\rho_0 = 0$.
\STATE $\forall r\in [m]:$ let distribution $\mathcal{D}^{(r)} =\NN\left(0, \frac{1}{2}\Delta_2^2 (\rho_r - \rho_{r-1})\right)$
\FOR{each round $r \in[m]$} 
    \STATE Initialize $Y^{(r)} = \{0\}^d$
    \FOR{each tracked count in preceding round $j\in S^{(r-1)}$}
        \STATE Draw fresh noise $\tilde{Z}^{(r)}_j\sim\mathcal{D}^{(r)}$\label{eff:fresh}
        \STATE Update aggregate noise $Z^{(r)}_j = \frac{\rho_{r-1}}{\rho_{r}}Z^{(r-1)}_j + \frac{1}{\rho_r}\tilde{Z}^{(r)}_j$
        \IF{$H(\X)_j + Z^{(r)}_j > \tau_r$}
            \STATE $Y^{(r)}_j = H(\X)_j + Z^{(r)}_j$
        \ENDIF
    \ENDFOR
    \\ {\bf \color{blueish} // Simulating Noise for zero counts}
    \STATE Let $\tilde{Z}^{(k)}\sim\mathcal{D}^{(k)}$ for $k\in[r]$ be random variables.
    \STATE Compute $p^{(r)} = \Pr\left[\sum_{k=1}^r \tilde{Z}^{(k)} > \tau_r \rho_r\ |\ \{\forall \ell\in[r-1] : \sum_{k=1}^\ell \tilde{Z}^{(k)} \leq \rho_\ell\tau_\ell \} \right]$
    \STATE Draw $q\sim\operatorname{Binomial}\left(d - \lvert S^{(r-1)}\rvert, p^{(r)}\right)$.
    \STATE Select a subset $Q\subseteq [d]\setminus S^{(r-1)}$ uniformly at random of size $q$.
    \FOR{each index $j \in Q$}\label{eff:bin}
        \STATE Initialize $Z^{(r)}_j =0$
        \FOR{$k\in[r-1]$}
            \STATE Draw fresh noise $\tilde{Z}^{(k)}_j \sim \mathcal{D}^{(k)}$ conditioned on $\tilde{Z}^{(k)}_j \leq \rho_k\tau_k - Z^{(r)}_j$
            \STATE Update aggregate noise $Z^{(r)}_j = Z^{(r)}_j + \tilde{Z}^{(k)}_j$
        \ENDFOR
        \STATE Draw fresh noise $\tilde{Z}^{(k)}_j \sim \mathcal{D}^{(k)}$ conditioned on $\tilde{Z}^{(k)}_j > \rho_k\tau_k - Z^{(r)}_j$
        \STATE Update aggregate noise $Z^{(r)}_j = Z^{(r)}_j + \tilde{Z}^{(k)}_j$
        \STATE $Y^{(r)}_j = H(\X)_j + Z^{(r)}_j$ {\color{blueish}\bf //  Guaranteed to be above the threshold.}
    \ENDFOR
    \STATE Set $S^{(r)} = S^{(r-1)}\cup Q$
    \STATE\textbf{output} private histogram $Y^{(r)}$
\ENDFOR
\end{algorithmic}
\end{algorithm}

We proceed to give proofs for \Cref{lemma:efficient-histogram-id}~and~\ref{lemma:sampled_gaussians}.
\begin{proof}[Proof sketch for \Cref{lemma:efficient-histogram-id}.]
    The claim directly follows an inductive argument.
    For the base case, note that the first iteration of \Cref{alg:efficient-gradual-release} is running the same routine as described by \Cref{lemma:smart-histogram-sampling}, and so $\hat{Y}^{(1)}$ must be identically distributed to $Y^{(1)}$.
    For our inductive hypothesis, assume that the subsequences $(Y^{(1)}, \dots, Y^{(k)})$ and $(\hat{Y}^{(1)}, \dots, \hat{Y}^{(k)})$ are identically distributed.
    Note that for the $(k+1)$\textsuperscript{st} release, \Cref{alg:efficient-gradual-release} handles all the true nonzero counts and every zero count that has ever exceeded the threshold in a past round in the same way.
    These counts in $\hat{Y}^{(k+1)}$ and $Y^{(k+1)}$ clearly have the same distribution.
    
    Now, note that each of the remaining zero counts have up to and including the $k$\textsuperscript{th} round been reported as zero in each round, and so their probability of exceeding the threshold, $p^{(k+1)}$, should be equal across all of them.
    
    The sampling performed by \Cref{alg:efficient-gradual-release} is structured in the same manner as \Cref{lemma:smart-histogram-sampling}, but with sampling probability $p^{(k+1)}$, and the noise terms added for any zero count exceeding the threshold, are different.
    For the distributions to match, the probability of exceeding the threshold and the noise distribution for an element chosen to exceed the threshold are more complex.
    The probability of exceeding the threshold is conditioned on prefixes of Gaussians, which correctly simulates the probability of not exceeding the threshold at any prior round until first doing so in the $(k+1)$\textsuperscript{st} one.
    The noise added once selected to exceed the threshold is a sum of truncated Gaussians, which simulate the same event.
    As the principle is built on the same logic as \Cref{lemma:smart-histogram-sampling}, we have that $(Y^{(1)}, \dots, Y^{(k+1)})$ and $(\hat{Y}^{(1)}, \dots, \hat{Y}^{(k+1)})$ are identically distributed, and so the claim is true by induction.
\end{proof}
\begin{proof}[Proof of \Cref{lemma:sampled_gaussians}]
    Note that each of the $k$ non-zero true counts will, in each round, have fresh noise added on line~7.
    If a non-zero count is selected by the binomial sampling in the $k$\textsuperscript{th} round, the for-loop starting on line~17 will result in sampling $k$ noise terms, after which for future rounds it will get treated as a non-zero count.
    It follows that non-zero entries contribute $km$ samples, and zero counts exceeding the threshold contribute $cm$ samples.\end{proof}
    
\section{Definitions}\label{sec:preliminaries}

We begin with the most common version of differential privacy.
\begin{definition}[$(\epsilon, \delta)$-Differential Privacy~\cite{dwork2014algorithmic}]
    A randomized algorithm $\cM : \mathcal{X}\to\mathcal{Y}$ is ($\varepsilon, \delta)$-differentially private if for all
    $S\subseteq\mathsf{Range}(\cM)$ and all pairs of neighboring inputs $\X, \X' \in \cX$, it holds that
    \begin{equation*}
        \Pr[\cM(\X)\in S] \leq \exp(\epsilon)\Pr[\cM(\X')\in S] + \delta\,,
    \end{equation*}
    where $(\epsilon, 0)$-DP is referred to as $\varepsilon$-DP.
\end{definition}

Zero-Concentrated Differential Privacy (zCDP) is a notion of differential privacy that provides a simple but accurate analysis of privacy loss, particularly under composition.
\begin{definition}[\citet{bun_steinke_2016}, $\rho$-zCDP]
Let $\rho > 0$.
An algorithm ${\cM : \cX \rightarrow \cY}$ satisfies $\rho$-zCDP, 
if for all $\alpha > 1$ and all pairs of neighboring inputs $\X, \X' \in \cX$, it holds that 
\begin{equation*}
    \begin{array}{c}
        D_\alpha \left( 
            \cM \left( \X \right) || \cM \left( \X' \right) 
        \right)
        \le \rho \alpha, 
    \end{array}
\end{equation*}
where 
$
    D_\alpha \left( 
        \cM \left( \X \right) || \cM \left(\X' \right)\right)
$ 
denotes the $\alpha$-Rényi divergence between two output distributions of $\cM(\X)$ and $\cM(\X')$.
\end{definition}

\begin{lemma}[\citet{bun_steinke_2016}, Composition]\label{fact:composition-zdp}
    If $\cM_1(\X)$ and $\cM_2(\X)$ satisfy $\rho_1$-zCDP and $\rho_2$-zCDP, respectively,
    then $(\cM_1(\X), \cM_2(\X))$ satisfies $(\rho_1 + \rho_2)$-zCDP.
\end{lemma}

\begin{lemma}[\citet{bun_steinke_2016}, Gaussian Mechanism]
Let $f : \mathcal{X} \to \R^d$ be a query with $\ell_2$-sensitivity $\Delta_2$.
Consider the mechanism $\mathcal{G}_{f, \rho} : \mathcal{X} \to \R^d$ that, on private input $\X$, releases a sample from $f(\X) + \NN\left(0, \frac{\Delta_2^2}{2\rho}\right)^d$.
Then $\mathcal{G}_{f, \rho}$ satisfies $\rho$-zCDP.
\end{lemma}

One way to uniquely describe probability distributions is via their \emph{Characteristic Functions}. 
\begin{definition}[Characteristic Function]\label{def:CF}
The characteristic function of a random variable $X$ is defined as $\varphi_X(t) = \EX[e^{itX}]$.
\end{definition}

For proving \cref{lem:lap-conv,lem:exp-conv}, we will use a nice property of CFs for the convolution of two random variables:
\begin{lemma}[Convolution of Characteristic Functions]
Let $X,Y$ be two independent RVs with CFs $\varphi_X(t)$ and $\varphi_Y(t)$ respectively, then $\varphi_{X+Y}(t) = \varphi_X(t) \cdot \varphi_Y(t)$.
\end{lemma}
\begin{proof}
Furthermore, by linearity of expectation, we have for the sum of $X+Y$:
\begin{equation*}
    \varphi_{X+Y}(t) = \EX[e^{it(X+Y)}] = \EX[e^{itX}] + \EX[e^{itY}] =\varphi_X(t)\cdot \varphi_X(t)\,.\qedhere
\end{equation*}
\end{proof}
 
\section{The Poisson Mechanism}\label{sec:poisson}

We are not aware of any explicit statements in the literature on the privacy guarantees obtained by adding Poisson distributed noise to a $d$-dimensional vector of integers.
However, since the Poisson distribution is the limiting distribution of binomial distributions with the same mean $\lambda = N p$, where $N$ is the number of trials, such bounds can be derived from existing bounds on the binomial mechanism.
For the sake of completeness, we include such statements based on the following theorem from~\cite{agarwal2018cpsgd}:

\begin{theorem}[\citet{agarwal2018cpsgd}]
For any \(\delta\), parameters \(N, p\) and sensitivity bounds \(\Delta_1, \Delta_2, \Delta_\infty\) such that
\[
Np(1 - p) \geq \max \left( 23 \log \left( \frac{10d}{\delta} \right), \, 2\Delta_\infty \right),
\]
the $d$-dimensional Binomial mechanism is \((\varepsilon, \delta)\)-differentially private for
\[
\varepsilon =
\frac{\Delta_2 \sqrt{2 \log \frac{1.25}{\delta}}}{\sqrt{N p (1 - p)}}
+ \frac{\Delta_2 c_p \sqrt{\log \frac{10}{\delta}} + \Delta_1 b_p}{N p (1 - p) (1 - \delta/10)}
+ \frac{\frac{2}{3} \Delta_\infty \log \frac{1.25}{\delta} + \Delta_\infty d_p \log \frac{20d}{\delta} \log \frac{10}{\delta}}{N p (1 - p)}.
\]
where
\[
d_p \triangleq \frac{4}{3} \cdot \left( p^2 + (1 - p)^2 \right), \quad
b_p \triangleq \frac{2(p^2 + (1 - p)^2)}{3} + (1 - 2p), \quad
c_p \triangleq \sqrt{2} \left( 3p^3 + 3(1 - p)^3 + 2p^2 + 2(1 - p)^2 \right).
\]
\end{theorem}

Setting $p = \lambda/N$ and considering the limiting bound when $N\rightarrow\infty$ we get:

\begin{theorem}[Privacy guarantees of the Poisson Mechanism]
The $d$-dimensional Poisson mechanism with parameter
\( \lambda > \max( 23 \log(10d/\delta),\, 2\Delta_\infty) \),
is \((\varepsilon, \delta)\)-differentially private with
\[
\varepsilon =
\frac{\Delta_2 \sqrt{2 \log \frac{1.25}{\delta}}}{\sqrt{\lambda}}
+ \frac{5\sqrt{2}\,\Delta_2 \sqrt{\log \frac{10}{\delta}} + \frac{5}{3}\Delta_1}{\lambda (1-\delta/10)}
+ \frac{\frac{2}{3}\Delta_\infty \log \frac{1.25}{\delta} + \frac{4}{3}\Delta_\infty \log \frac{20d}{\delta}\,\log \frac{10}{\delta}}{\lambda} \,.
\]
\end{theorem}

A simpler expression can be derived for unit sensitivities by relaxing the constants and assuming that \(\delta\) is not too large:

\begin{corollary}[Simplified upper bound with unit sensitivities]
    Assume that \(\Delta_1=\Delta_2=\Delta_\infty=1\). Then for \(\delta < 1/100\), the $d$-dimensional Poisson mechanism with parameter \(\lambda > 23 \log(10d/\delta)\) is \((\varepsilon, \delta)\)-differentially private for
    \[
    \varepsilon = \frac{\sqrt{2\log\frac{1.25}{\delta}}}{\sqrt{\lambda}} 
    + \frac{2\,\log\frac{20d}{\delta}\,\log\frac{10}{\delta}}{\lambda}\,.
    \]
    \end{corollary}

We will need the following lemma that determines the sampling step of private lossless multiple release for the Poisson mechanism:

\begin{lemma}\label{lemma:conditional-poisson}
    Let $X_1 \sim \operatorname{Poi}(\lambda_1)$ and $X_2 \sim \operatorname{Poi}(\lambda_2)$ be independent Poisson random variables. 
    Then, for any nonnegative integer $k$, the conditional distribution of $X_1$ given $X_1+X_2 = k$ is
    \begin{align*}
        X_1\mid (X_1+X_2=k) \sim \operatorname{Binomial}\left(k, \frac{\lambda_1}{\lambda_1+\lambda_2}\right)\,.
    \end{align*}
    \end{lemma}
    
    \begin{proof}
    Since \( X_1 \) and \( X_2 \) are independent, their joint probability mass function is for all $x = 0,1, \cdots, k$
    \begin{align*}
    P[X_1=x,\,X_2=k-x] = e^{-(\lambda_1+\lambda_2)}\frac{\lambda_1^x}{x!}\frac{\lambda_2^{\,k-x}}{(k-x)!}\,.
    \end{align*}
    Moreover, the sum \( X_1+X_2 \) is Poisson with parameter \( \lambda_1+\lambda_2 \), so that
    \[
    \Pr[X_1+X_2=k] = e^{-(\lambda_1+\lambda_2)}\frac{(\lambda_1+\lambda_2)^k}{k!}\,.
    \]
    Thus, by the definition of conditional probability,
    \begin{align*}
    \Pr[X_1=x \mid X_1+X_2=k] &= \frac{\Pr[X_1=x,\,X_2=k-x]}{\Pr[X_1+X_2=k]} \\
    &= \dfrac{e^{-(\lambda_1+\lambda_2)}\dfrac{\lambda_1^x}{x!}\dfrac{\lambda_2^{\,k-x}}{(k-x)!}}{e^{-(\lambda_1+\lambda_2)}\dfrac{(\lambda_1+\lambda_2)^k}{k!}} \\
    &= \binom{k}{x}\left(\dfrac{\lambda_1}{\lambda_1+\lambda_2}\right)^x \left(\dfrac{\lambda_2}{\lambda_1+\lambda_2}\right)^{k-x}.
    \end{align*}
    This is precisely the probability mass function of a Binomial random variable with parameters \( k \) and \( \frac{\lambda_1}{\lambda_1+\lambda_2} \). 
    \end{proof}

 \section{The Laplace Mechanism}\label{sec:laplace}

\citet{koufogiannis_2016} showed how to do lossless \emph{gradual} releases for the Laplace mechanism and supports either tightening the privacy guarantees or loosening them. 
We will now strengthen this result by exactly showing how to do them in arbitrary order, similar to what was done in \cref{sec:lossless-multiple-release} for the Gaussian mechanism.
We will first show that the Laplace distribution satisfies \cref{def:decomposable}.
This already implies the existence of an algorithm supporting gradual lossless release via \cref{lemma:additive_noise_means_lossless}.
After, we will derive how to sample a new release with scale parameter $b$ given two distinct releases with scaling $b_2 < b$ and $b < b_1$.

\begin{definition}[Laplace distribution]\label{def:lap}
The \emph{zero-centered Laplace distribution} $\lap(0, b)$ with scale parameter $b > 0$ has probability density function $f_b(x) = \frac{1}{2b} \exp( -|x|/b )$ for all $x \in \mathbb{R}$.
\end{definition}

\begin{lemma}[{\citet{kotzLaplaceDistributionGeneralizations2001}} Characteristic function of Laplace]\label{lem:CFLap}
Let $X$ be Laplace random variable with probability density function as in \cref{def:lap}, then for all $t\in \R$, the characteristic function of $X$ is
\[ 
\varphi_X(t) = \EX[e^{itX}] = \int\limits_{-\infty}^{\infty}e^{i t x}\frac{1}{2b}e^{-|x|/b} dx = \frac{1}{1+b^2t^2}\,.
\]
\end{lemma}
\begin{proof}
    By a simple integration:
    \begin{align*}
        \EX[e^{itX}] &= \frac{1}{2b} \int\limits_{-\infty}^{\infty}e^{itx -|x|/b} dx
        = \frac{1}{2b}\left(\int\limits_{-\infty}^{0}e^{(1/b + ti)x} dx
        + \int\limits_{0}^\infty e^{(-1/b + ti)x}dx \right)\\
        &= \dfrac{1}{2b}\left(\dfrac{1}{1/b + ti}+ \dfrac{1}{1/b - ti}\right) = \dfrac{1}{1+b^2t^2}\,.\qedhere
    \end{align*}
\end{proof}
We will now show that the zero-centered Laplace distribution with scale parameter $b$ satisfies convolution preorder (\cref{def:decomposable}).
Note that a larger value of $b$ corresponds to a more private release by definition.
Building on this, we will condition on the two closest releases to create a new one in the middle, as done in \cref{lemma:gaussian_seq_ooo} for the Gaussian.
The following fact was already shown in~\cite{koufogiannis_2016}, but we give a proof here for completeness.
\begin{claim}[Convolution preorder: Laplace]\label{lem:lap-cpo}
    Fix $b_2 < b_1 \in \R^+$ let $X \sim \lap(0, b_2)$ and draw 
    \begin{align*}
        W = \begin{cases}
        0 & \text{with probability } b_2^2/b_1^2\\
        \lap(0, b_1) & \text{otherwise}
    \end{cases},  \quad \text{then~} X + W \sim \lap(0, b_1)\,.
    \end{align*}
\end{claim}
\begin{proof}
We know that the characteristic function of $X$ is $\varphi_{X}(t) = \frac{1}{1+b_2^2t^2}$ and for the convolution $\varphi_{X+W} = \frac{1}{1+b_1^2t^2}$.
Because of independence, the convolution is defined for all $t$ as:
\begin{align*}
\varphi_{X}(t)\varphi_{W}(t) &= \varphi_{X+W}(t)\\
\Leftrightarrow \varphi_{W}(t) &= \frac{1+b_2^2t^2}{1+b_1^2t^2} = \frac{b_1^2 + b_1^2b_2^2t^2}{b_1^2(1+b_1^2t^2)}
= \frac{b_2^2(1+b_1^2t^2) + b_1^2-b_2^2}{b_1^2(1+b_1^2t^2)}
= \frac{b_2^2}{b_1^2}+ \left(1 - \frac{b_2^2}{b_1^2}\right) \cdot \frac{1}{1+b_1^2t^2}\,.
\end{align*}

The last expression encodes the convex combination of the claimed mixture distribution because $\frac{1}{1+b_1^2t^2}$ is again the characteristic function of a zero-centered Laplace distribution with scaling parameter $b_1$.
\end{proof}

We next prove a simple result about the convolution of two Laplace distributions (compare also eq., 2.3.23 of \citet{kotzLaplaceDistributionGeneralizations2001}).
\begin{lemma}[Convolution]\label{lem:lap-conv}
    For two fixed scaling parameters $b_1 \neq b_2 \in \R^+$, let $X_1 \sim \lap(0, b_1)$ and $X_2 \sim \lap(0, b_2)$. Then the density of $X_1 + X_2$ is given by
    \begin{align*}
        (f_{b_1} \ast f_{b_2})(t)
        = \dfrac{1}{2(b_1^2 - b_2^2)}\left( b_1 e^{-|t|/b_1} - b_2 e^{-|t|/b_2}\right)\,.
    \end{align*}
\end{lemma}
\begin{proof}
Assume $t\geq 0$ and note that the other case follows by symmetry.
We compute the density of the convolution by a straightforward integration:
\begin{align*}
    (f_{b_1} \ast f_{b_2})(t)  &= \int_{-\infty}^{\infty}f_{b_1}(x)f_{b_2}(t-x)dx = \dfrac{1}{4b_1b_2}\int_{-\infty}^{\infty}\exp\left(-\frac{|x|}{b_1} - \frac{|t-x|}{b_2}\right)dx\\
    &= \dfrac{1}{4b_1b_2}\cdot \left(\int_{-\infty}^{0} e^{\frac{x}{b_1} - \frac{t-x}{b_2}} dx
    + \int_{0}^{t} e^{-\frac{x}{b_1} - \frac{t-x}{b_2}} dx
    + \int_{t}^{\infty}e^{-\frac{x}{b_1} - \frac{x-t}{b_2}} dx\right)\\
    &=\dfrac{1}{4b_1b_2}\cdot \bigg(\left[\dfrac{\exp\left(x/b_1 - (t-x)/b_2\right)}{b_1^{-1} +b_2^{-1}}\right]_{-\infty}^{0}\\
    &\qquad+ \left[\dfrac{\exp(-x/b_1 - (t-x)/b_2)}{b_2^{-1} - b_1^{-1}}\right]_{0}^{t}
    + \left[\dfrac{\exp(-x/b_1 - (x-t)/b_2)}{-b_2^{-1} - b_1^{-1}}\right]_{t}^{\infty}\bigg)\\
    &= \dfrac{1}{4} \cdot \left(\dfrac{\exp(-t/b_2)}{b_1 + b_2}
    + \dfrac{\exp(-t/b_1) - \exp(-t/b_2)}{b_1 - b_2}
    + \dfrac{\exp(-t/b_1)}{b_1 + b_2}\right)\\
    &= \frac{1}{2(b_1^2 - b_2^2)}\left( b_1 e^{-t/b_1} - b_2 e^{-t/b_2}\right)\,.\qedhere
\end{align*}
\end{proof}

Now, we are ready to show that the Laplace mechanism can be implemented with multiple releases.
\begin{lemma}[Multiple release Laplace]\label{lem:lap-mult}
For fixed $0 < b_2 < b < b_1$, let $\mu_1 = b^2/b_1^2$ and $\mu_2 = b_2^2/b^2$ and $D_1 \sim \ber(\mu_1)$ and $D_2 \sim \ber(\mu_2)$.
Furthermore, let
\begin{align*}
X_1 =
\begin{cases}0 & \text{if} \quad D_1 = 1 \\ 
\lap(0, b) & \text{otherwise }\end{cases}
\quad \text{and} \quad  
X_2 = 
\begin{cases}0 & \text{if} \quad D_2 = 1 \\ \lap(0, b_2) & \text{otherwise } \end{cases}.
\end{align*}
Then we have that $\Pr[X_1 = 0 \mid X_1 + X_2 = 0] = 1$, and
for every real number $k\neq 0$ the distribution of $X_1$ conditioned on $X_1 + X_2 = k$ is:
\begin{align}\label{alg:sample}
X_1\mid{(X_1 + X_2 = k)} \sim \begin{cases}
    0 & \text{with probability } \mu_1\cdot(1-\mu_2)\cdot \dfrac{\exp(-|k|/b_2)}{2b_2\cdot f_{X_1+X_2}(k)} \\[2ex]
    k & \text{with probability } (1-\mu_1)\cdot \mu_2\cdot\dfrac{\exp(-|k|/b)}{2b\cdot f_{X_1+X_2}(k)}\\[2ex]
    H(b, b_2, k) & \text{with probability } (1-\mu_1)\cdot (1- \mu_2)\cdot \dfrac{(f_b \ast f_{b_2})(k)}{f_{X_1+X_2}(k)}\\
    \end{cases} 
\end{align}    
where $0$ and $k$ denote the constant distributions, $H(b, b_2, k)$ is the probability distribution with probability density function $h(x) = \dfrac{f_b(x)f_{b_2}(k-x)}{(f_b \ast f_{b_2})(k)}$, 
\begin{align*}
f_{X_1 + X_2}(k) &= \frac{\mu_1(1-\mu_2}{2b_2} e^{-|k|/b_2} + \frac{(1-\mu_1)\mu_2}{2b} e^{-|k|/b}
  \frac{(1-\mu_1)(1-\mu_2)}{2(b^2-b_2^2)} \left(be^{-|k|/b}- b_2 e^{-|k|/b_2}\right) + \mu_1\mu_2\delta_0(x) ,\text{ and}\\
(f_b \ast f_{b_2})(k) &= \frac{1}{2(b_1^2 - b_2^2)}\left( b_1 e^{-t/b_1} - b_2 e^{-t/b_2}\right)\,,\end{align*}
where $\delta_0$ is the Dirac delta function.
\end{lemma}

\begin{proof}
First we consider $\Pr[X_1 = 0| X_1 + X_2 = 0]$, arguing that
\begin{align*}
    \Pr[X_1 = 0| X_1 + X_2 = 0]
    = \frac{\Pr[X_1 = 0]\Pr[X_2=0]}{\Pr[X_1 + X_2 = 0]} = \frac{\mu_1\mu_2}{\mu_1\mu_2 + 0}= 1 \,.
\end{align*}
To see why the first equality holds split $\Pr[X_1 + X_2 = 0]$ into each of the four combinations of discrete/continuous for $X_1, X_2$.
The only non-zero contribution to the probability mass comes from the discrete/discrete case, as the remaining cases contribute mass proportional to the probability that a continuous distribution assumes an exact value, which is zero.

Next we turn to the distribution of $X_1$ conditioned on $X_1 + X_2 = k$ where $k\neq 0$.
Denote by $f_b$ the probability density function of $\lap(0, b)$.
Note that the mixture densities become
\begin{align*}
    f_{X_1}(x) &= \mu_1 \delta_0(x) + (1-\mu_1) f_b(x)\,,\\
    f_{X_2}(x) &= \mu_2 \delta_0(x) + (1-\mu_2)f_{b_2}(x)\,.
\end{align*}
Before analyzing the different cases separately, we compute the convolution $X_1 + X_2$.

\allbold{Convolution $f_{X_1 + X_2}$.}

Note that have to take care of a subtle technicality: The case that both $X_1$ and $X_2$ are zero from the discrete part can only happen when we condition on $X_1 + X_2 = 0$.
We can now give the density of the convolution $X_1 + X_2$ at any real point $x$:
\begin{align*}
f_{X_1 + X_2}(x) &=  \sum_{(d_1,d_2) \in \{0,1\}^2} f_{X_1+X_2|D_1 = d_1, D_2 = d_2}(x) \cdot \Pr[D_1 = d_1 \wedge D_2 = d_2]\\
&=\mu_1 (1-\mu_2) \cdot  f_{b_2}(x) + (1-\mu_1) \mu_2 \cdot f_{b}(x)
+ (1-\mu_1)(1-\mu_2) \cdot (f_{b} \ast f_{b_2})(x)
+ \mu_1\mu_2\delta_0(x)\\
&=\dfrac{\mu_1(1-\mu_2)}{2b_2}e^{-|x|/b_2}
+ \frac{(1-\mu_1)\mu_2}{2b}e^{-|x|/b}
+ (1-\mu_1)(1-\mu_2) \frac{1}{2(b^2-b_2^2)}\left(b e^{-|x|/b}
- b_2 e^{-|x|/b_2}\right)\\
&\qquad + \mu_1\mu_2\delta_0(x)\,,
\end{align*}
where the third term follows from \Cref{lem:lap-conv}. 
Note that the last term only contributes when $x=0$.

We can now show how the sampling procedure in the claim is justified. 
We first assume $k \neq 0$ and analyze three possible cases how $X_1 + X_2$ is built up: 
Either both of them are drawn from the (continuous) Laplace distribution or exactly one.
(The case where both are from their respective discrete parts can only happen if $k = 0$, analyzed above.)
\begin{caseof}
\case{$X_1 = 0$ and $X_2 = k$.}
$X_2 = k$ is necessarily from its continuous part.
By the definition of conditional probability, we have for $k \neq 0$:
\begin{align*}
    \Pr[X_1 = 0 | X_1 + X_2 = k, k \neq 0] 
    &= \frac{\Pr[X_1 = 0 \wedge X_1+X_2=k]}{\Pr[X_1 + X_2 = k]} 
    = \frac{\Pr[X_1 = 0]\Pr[X_2=k]}{\Pr[X_1 + X_2 = k]}\\
    &= \mu_1\frac{f_{X_2}(k) }{f_{X_1 + X_2}(k)}
    = \frac{\mu_1(1-\mu_2)}{2b_2 f_{X_1+X_2}(k)} \cdot e^{-|k| / b_2} \coloneqq p_1\,.
\end{align*}
For the third equality, we simply used definition of a probability density function via its limit:
\[
\lim_{\Delta\to 0^+} \frac{\Pr\left[X_2\in[k-\Delta, k+\Delta]\right]}{\Pr\left[X_1+X_2\in[k-\Delta, k+\Delta]\right]}
= \frac{f_{X_2}(k)}{f_{X_1+X_2}(k)}\,.
\]

\case{$X_1 = k$ and $X_2 = 0$}
Now assume the flipped case.
By a similar argument:
\[
\Pr[X_1 = k | X_1 + X_2 = k] = \frac{\Pr[X_1 = k, X_1 + X_2 = k]}{\Pr[X_1 + X_2 = k]}  = \frac{\Pr[X_1 = k]\Pr[X_2 = 0]}{\Pr[X_1 + X_2 = k]}
=\frac{(1-\mu_1)\mu_2}{2bf_{X_1+X_2}(k)} \cdot e^{-|k|/b} \coloneqq p_2 \,.
\]

\case{$X_1 \neq 0, X_2 \neq 0, X_1 + X_2 = k$} 
In the remaining case, both $X_1$ and $X_2$ are independently sampled from continuous Laplace distributions with probability density functions $f_b(x)$ and $f_{b_2}(x)$.
Therefore, with the remaining probability $1-(p_1 + p_2)$, we know that $X_1$ is sampled according to the following conditional probability density function:
\[
f_{X_1}(x)\vert_{X_1 + X_2 = k} =
\frac{f_{X_1, X_1+X_2}(x,k)}{(f_b \ast f_{b_2})(k)}
=\frac{f_b(x)f_{b_2}(k-x)}{(f_{b}\ast f_{b_2})(k)}\,,
\]
where the last line follows from the trivial identity $X_1 + X_2 = k \Leftrightarrow X_2 = k - X_1$.
This is a valid probability density function because $f$ is trivially non-negative due to its parts being non-negative and furthermore
\begin{align*}
\int_{-\infty}^{\infty}\frac{f_b(x)f_{b_2}(x-k)}{(f_{b} \ast f_{b_2})(k)} dx
= \frac{1}{(f_b \ast f_{b_2})(k)}\int_{-\infty}^{\infty}f_b(x)f_{b_2}(k-x) dx
= \frac{(f_b \ast f_{b_2})(k)}{(f_{b}\ast f_{b_2})(k)} = 1\,.
\end{align*}

\end{caseof}
What is left is to verify that the probabilities in \cref{alg:sample} indeed add up to one for $k\ne 0$:
\begin{align*}
      &\frac{(1-\mu_1)(1- \mu_2) (f_b \ast f_{b_2})(k)}{f_{X_1+X_2}(k)} + \frac{  
     (1-\mu_1)\mu_2\exp(-|k|/b)}{2b f_{X_1+X_2}(k)}  +\frac{\mu_1(1-\mu_2) \exp(-|k|/b_2)}{2b_2 f_{X_1+X_2}(k)} \\
    &  = \frac{1}{f_{X_1+X_2}(k)}
    \underbrace{\left((1-\mu_1) (1- \mu_2) (f_b \ast f_{b_2})(k) + \frac{(1-\mu_1)\mu_2 \exp(-|k|/b)}{2b} + \frac{\mu_1(1-\mu_2) \exp(-|k|/b_2)}{2b_2}e^{-\lambda_2 |k|} \right)}_{f_{X_1 + X_2}(k) \ \text{for} \ k \neq 0} = 1
    \end{align*}
\end{proof}

\subsection{Showing Convolution Preorder for Exponential Noise}\label{sec:exponential}

The exponential distribution is closely related to the Laplace distribution, but gives poor differential privacy guarantees.
Nevertheless, it still serves as a building block for some private mechanisms, e.g., Report-Noisy-Max~\cite{Ding2021}.
We show next that it also satisfies a convolution preorder.

\begin{definition}[Exponential distribution]\label{def:Exp}
The \emph{exponential distribution} $\operatorname{Exp}(\lambda)$ with rate parameter $\lambda > 0$ has probability density function $f_\lambda(x) = \lambda \exp( -\lambda x)$ for all $x \in \mathbb{R}_+$.
Furthermore, its characteristic function is given by $\varphi_X(t) = \EX[e^{itX}] = \frac{\lambda}{\lambda- it}$
\end{definition}

\begin{claim}[Convolution preorder: Exponential distribution]\label{lem:exp-conv}
    Fix $\lambda_1 < \lambda_2  \in \R_+$, let $X \sim \operatorname{Exp}(\lambda_2)$ and draw 
    \begin{align*}
    W = \begin{cases}
        0 & \text{with probability } \lambda_1/\lambda_2 \\
        \operatorname{Exp}( \lambda_1) & \text{otherwise} 
    \end{cases},  \quad \text{then~} X + W \sim \operatorname{Exp}(\lambda_1)\,.
    \end{align*}
\end{claim}
\begin{proof}
    Using the same trick as in the proof of \Cref{lem:lap-cpo}, we have that
\begin{align*}
    \varphi_{W}(t) &= \frac{\varphi_{X + W}(t)}{\varphi_{X}(t)}
    = \frac{\lambda_1}{\lambda_2}\cdot\frac{\lambda_2 - it}{\lambda_1 - it}
    = \frac{\lambda_1}{\lambda_2}\left( 1 + \frac{\lambda_2 - \lambda_1}{\lambda_1 - it}\right)
    = \frac{\lambda_1}{\lambda_2} + \left( 1 - \frac{\lambda_1}{\lambda_2} \right) \varphi_{X+W}(t)
\end{align*}
where the final expression is the characteristic function of the claimed mixture distribution.\end{proof} 

\end{document}